\documentclass[a4paper]{article}%
\usepackage{amsmath}
\usepackage{amsfonts}
\usepackage{amssymb}
\usepackage{graphicx,natbib}%
\usepackage[caption=false]{subfig}
\providecommand{\U}[1]{\protect\rule{.1in}{.1in}}
\newtheorem{theorem}{Theorem}

\newtheorem{definition}[theorem]{Definition}

\newenvironment{proof}[1][Proof]{\noindent \textbf{#1.} }{\  \rule{0.5em}{0.5em}}
\input{tcilatex}

\begin{document}

\title{\textbf{Robust Wald-type tests for non-homogeneous observations
based on minimum density power divergence estimator 
}}
\author{Ayanendranath Basu$^{1\ast}$, Abhik Ghosh$^{1}$, Nirian Martin$^{2}$ and Leandro Pardo$^{2}$\\
	$^1$ Indian Statistical Institute, Kolkata, India
	\\ $^2$ Complutense University, Madrid, Spain
%	\\ $^{*}$Corresponding author; Email: abhianik@gmail.com. 
\\ $^{\ast}$Corresponding author; Email: ayanbasu@isical.ac.in 
}
\date{}
\maketitle

\begin{abstract}
This paper considers the problem of robust hypothesis testing under non-identically distributed data.
We propose Wald-type tests for both simple and composite hypothesis for independent but non-homogeneous observations
based on the robust minimum density power divergence estimator of the common underlying parameter.
Asymptotic and theoretical robustness properties of the proposed tests have been discussed.
Application to the problem of testing the general linear hypothesis in a generalized linear model with fixed-design 
has been considered in detail with specific illustrations for its special cases under 
normal and Poisson distributions. 
\end{abstract}

%\bigskip\bigskip

%\noindent\underline{\textbf{AMS 2001 Subject Classification}}\textbf{: }
%

\section{Introduction}\label{sec1}

Suppose that the parametric model for the sample $Y_{1},...,Y_{n}$ asserts
that the distribution of $Y_{i}$ is $F_{i,\boldsymbol{\theta}}$,
$\boldsymbol{\theta}\in\Theta$ $\subset\mathbb{R}^{p}$. We shall denote by
$f_{i,\boldsymbol{\theta}}$ the probability density function associated with
$F_{i,\boldsymbol{\theta}}$ with respect to a convenient $\sigma$-finite
measure, for $i=1,...,n$. This situation is very common for statistical modeling in many applications and 
an important example is the generalized linear model (GLM) with fixed design set-up.

The maximum likelihood score equation for these independently and non
identically distributed data, $Y_{1},...,Y_{n}$, is given by
\begin{equation}%
%TCIMACRO{\sum \limits_{i=1}^{n}}%
%BeginExpansion
{\displaystyle\sum\limits_{i=1}^{n}}
%EndExpansion
\boldsymbol{u}_{i,\boldsymbol{\theta}}\left(  y_{i}\right)  =\boldsymbol{0},
\label{a.1}%
\end{equation}
with $\boldsymbol{u}_{i,\boldsymbol{\theta}}\left(  y_{i}\right)
=\frac{\partial}{\partial\boldsymbol{\theta}}\log f_{i,\boldsymbol{\theta}%
}(y_i)$. It is well-known that the maximum likelihood estimator,
$\widehat{\boldsymbol{\theta}}$, obtained as the solution of the system of equation
(\ref{a.1}), has serious problems of robustness. For this reason 
statistical solutions for several special cases like
linear regression with normal errors \citep{Huber:1983, Muller:1998} 
as well as some general cases \citep{Beran:1982,Ghosh/Basu:2013} have been considered in the literature. 
In this paper, we shall follow the approach presented by \cite{Ghosh/Basu:2013} in order to
propose some robust Wald-type tests. In the cited approach given by \cite{Ghosh/Basu:2013}
a robust estimator was introduced based on the density power divergence (DPD) measure; 
details about this family of divergence measures can be found in  \cite{Basu/etc:1998,Basu/etc:2011}. 
This estimator, called the minimum density power divergence estimator (MDPDE) 
for non-homogeneous observations, is obtained as the
solution of the system of equations
\begin{equation}%
%TCIMACRO{\sum \limits_{i=1}^{n}}%
%BeginExpansion
{\displaystyle\sum\limits_{i=1}^{n}}
%EndExpansion
\left(  f_{i,\boldsymbol{\theta}}^{\tau}(Y_{i})\boldsymbol{u}%
_{i,\boldsymbol{\theta}}\left(  Y_{i}\right)  -%
%TCIMACRO{\dint }%
%BeginExpansion
{\displaystyle\int}
%EndExpansion
f_{i,\boldsymbol{\theta}}^{\tau+1}(y)\boldsymbol{u}_{i,\boldsymbol{\theta}%
}\left(  y\right)  dy\right)  =\boldsymbol{0},\text{ }\tau>0. \label{a.2}%
\end{equation}
Note that in the limit as $\tau\rightarrow0$, the system of equations given
in (\ref{a.2}) tends to the system given in (\ref{a.1}).

In \cite{Ghosh/Basu:2013} it was established under some standard regularity conditions,
that the asymptotic distribution of the MDPDE for non-homogeneous observations, 
say $\widehat{\boldsymbol{\theta}}_{\tau}$, at the true model distribution 
$\{f_{i,\boldsymbol{\theta}_0} :  i=1,\ldots,n\},$ is given by
\begin{equation}
\boldsymbol{\Omega}_{n,\tau}^{-1/2}(\boldsymbol{\theta}_{0})
\boldsymbol{\Psi}_{n,\tau}(\boldsymbol{\theta}_{0})
\left[\sqrt{n}(\widehat{\boldsymbol{\theta}}_{\tau}-\boldsymbol{\theta}_{0})\right]
\overset{\mathcal{L}}{\underset{n\mathcal{\rightarrow}\infty}{\longrightarrow}}
\mathcal{N}\left(  \boldsymbol{0},\boldsymbol{I}_{p}\right), 
\label{a.250}%
\end{equation}
or equivalently, 
\begin{equation}
\sqrt{n}(\widehat{\boldsymbol{\theta}}_{\tau}-\boldsymbol{\theta}%
_{0})\overset{\mathcal{L}}{\underset{n\mathcal{\rightarrow}\infty
}{\longrightarrow}}\mathcal{N}\left(  \boldsymbol{0},\boldsymbol{\Sigma
}_{\tau}(\boldsymbol{\theta}_{0})\right)  , \label{a.25}%
\end{equation}
with
\begin{equation}
\boldsymbol{\Sigma}_{\tau}(\boldsymbol{\theta}_{0})=\lim_{n\rightarrow
\infty}\boldsymbol{\Psi}_{n,\tau}^{-1}(\boldsymbol{\theta}_{0}%
)\boldsymbol{\Omega}_{n,\tau}(\boldsymbol{\theta}_{0})\boldsymbol{\Psi
}_{n,\tau}^{-1}(\boldsymbol{\theta}_{0}), \label{a.3}%
\end{equation}
where we define
\begin{align}
\boldsymbol{\Psi}_{n,\tau}(\boldsymbol{\theta})  &  =\frac{1}{n}%
%TCIMACRO{\sum \limits_{i=1}^{n}}%
%BeginExpansion
{\displaystyle\sum\limits_{i=1}^{n}}
%EndExpansion
\boldsymbol{J}_{i,\tau}(\boldsymbol{\theta}),\label{a.4}\\
\mbox{with }~~~~~~~~~~~~~~~~~~~~~~~~~~~~~~
\boldsymbol{J}_{i,\tau}(\boldsymbol{\theta})  &  =%
%TCIMACRO{\dint }%
%BeginExpansion
{\displaystyle\int}
%EndExpansion
\boldsymbol{u}_{i,\boldsymbol{\theta}}\left(  y\right)  \boldsymbol{u}%
_{i,\boldsymbol{\theta}}^{T}\left(  y\right)  f_{i,\boldsymbol{\theta}%
}^{\tau+1}(y)dy,~~~~~~~~~~~~~~~~~~~~~~~~~~~~~~~~~~~~~~~~~~~~~~~~~~~~~~~~~~~~~~~~
\nonumber
\end{align}
and
\begin{align}
\boldsymbol{\Omega}_{n,\tau}(\boldsymbol{\theta})  &  =\frac{1}{n}%
%TCIMACRO{\sum \limits_{i=1}^{n}}%
%BeginExpansion
{\displaystyle\sum\limits_{i=1}^{n}}
%EndExpansion
\left(
%TCIMACRO{\dint }%
%BeginExpansion
{\displaystyle\int}
%EndExpansion
\boldsymbol{u}_{i,\boldsymbol{\theta}}\left(  y\right)  \boldsymbol{u}%
_{i,\boldsymbol{\theta}}^{T}\left(  y\right)  f_{i,\boldsymbol{\theta}%
}^{2\tau+1}(y_{i})dy-\boldsymbol{\xi}_{i,\tau}(\boldsymbol{\theta
})\boldsymbol{\xi}_{i,\tau}^{T}(\boldsymbol{\theta})\right)  ,
~~~~~~~~~~~~~~~~~~~\label{EQ:Omega_n}\\
\mbox{with }~~~~~~~~~~~~~~~~~~~~~~~~~~
\boldsymbol{\xi}_{i,\tau}(\boldsymbol{\theta})  &  =%
%TCIMACRO{\dint }%
%BeginExpansion
{\displaystyle\int}
%EndExpansion
\boldsymbol{u}_{i,\boldsymbol{\theta}}\left(  y\right)
f_{i,\boldsymbol{\theta}}^{\tau}(y)dy.\label{EQ:xi}
\end{align}
The required regularity conditions are listed in Appendix \ref{APP:cond}
for the sake of completeness and will be referred as the ``Ghosh-Basu Conditions" throughout the rest of the paper.

Motivated by the strong robustness properties of the Wald-type test statistics
based on MDPDEs \citep{Basu/etc:2016a,Basu/etc:2016b,Ghosh/etc:2016a,Ghosh/etc:2016b} 
in case of independently and identically distributed observations,
in this paper we shall introduce and study the corresponding MDPDE based 
Wald-type tests for independently but non identically distributed data. 
%based on minimum power divergence estimator (MDPDE). 
In particular, we will develop the asymptotic and theoretical robustness 
of these Wald-type tests for both simple and composite hypotheses,
along with  applications to the generalized linear models (GLMs) and its important subclasses.
It is important to note that there is no established robust hypothesis testing procedure 
under such general non-homogeneous data set-up except for one recent attempt by \cite{Ghosh/Basu:2015},
who developed the divergence based test statistics with DPD measure. 
However, the asymptotic null distribution  of their proposed test statistics 
is a linear combination of chi-square distributions occasionally limiting its application in complicated situations. 
On the contrary, our proposed Wald-type test statistics in this paper will be shown to have an ordinary chi-square distribution
along with all the other competitive properties and hence easier to implement in any practical applications.

The rest of the paper is organized as follows: 
In Section \ref{sec2}, we shall introduce the
Wald-type tests for testing simple null hypothesis as well as composite null
hypothesis and we study their asymptotic distributions under the null
hypotheses as well as alternative hypotheses. 
The robustness of these Wald-type tests 
%introduced in Section \ref{sec2} 
will be studied in Section \ref{sec3}. 
In Section \ref{sec4} the results are particularized to the GLM model. 
Some examples are studied in Section \ref{sec5} and 
the paper ends with some insightful discussions in Section \ref{sec7}.

\section{Wald-type tests under independent but non-homogeneous data}
\label{sec2}

In the following two sections we shall consider the simple null hypothesis as
well as composite null hypothesis versions of the Wald-type test statistics
for independently and non identically distributed data.

\subsection{Wald-type tests for simple null hypotheses}

Let $Y_{1},...,Y_{n}$ independently and non identically distributed data
according to the probability density function $f_{i,\boldsymbol{\theta}}$,
where $\boldsymbol{\theta}\in\Theta\subset\mathbb{R}^{p}.$ In this section we
define a family of Wald-type test statistics based on MDPDE for testing the
hypothesis%
\begin{equation}
H_{0}:\boldsymbol{\theta}=\boldsymbol{\theta}_{0}\text{ against }%
H_{1}:\boldsymbol{\theta}\neq\boldsymbol{\theta}_{0}, \label{a.6}%
\end{equation}
for a given $\boldsymbol{\theta}_0\in \Theta$, 
which will henceforth be referred to as the proposed Wald-type statistics.

\begin{definition}
Let $\widehat{\boldsymbol{\theta}}_{\tau}$ be the MDPDE of
$\boldsymbol{\theta}$. The family of proposed Wald-type test statistics for
testing the null hypothesis (\ref{a.6}) is given by%
\begin{equation}
W_{n}^0(\boldsymbol{\theta}_{0})=n(\widehat{\boldsymbol{\theta}}_{\tau
}-\boldsymbol{\theta}_{0})^{T}\boldsymbol{\Sigma}_{\tau}^{-1}%
(\boldsymbol{\theta}_{0})(\widehat{\boldsymbol{\theta}}_{\tau}%
-\boldsymbol{\theta}_{0}), \label{a.7}%
\end{equation}
where $\boldsymbol{\Sigma}_{\tau}(\boldsymbol{\theta}_{0})$ is as defined in (\ref{a.3}).
\end{definition}

The asymptotic distribution of $W_{n}^0(\boldsymbol{\theta}_{0})$ is presented
in the next theorem. The result follows easily from the asymptotic
distribution of the MDPDE considered in (\ref{a.25}) and so we omit the proof.
Throughout the rest of the paper, for all the theoretical results, 
we will assume that the Ghosh-Basu conditions hold and 
$\boldsymbol{\Sigma}_{\tau}(\boldsymbol{\theta})$ is continuous in $\boldsymbol{\theta}\in \Theta$.

\begin{theorem}\label{THM:null}
The asymptotic distribution, under the null hypothesis considered in
(\ref{a.6}), of the proposed Wald-type test statistics given in (\ref{a.7}) is $\chi_p^2$,
a chi-square distribution with $p$ degrees of freedom.
\end{theorem}

In the next Theorem we are going to present a result that will be important in
order to get an approximation to the power function of the proposed Wald-type
test statistic given in (\ref{a.7}) because in many practical situation is not
possible to get a simple expression for the exact power function.

\begin{theorem}
Let $\boldsymbol{\theta}^{\ast}$ be the the true value of parameter with
$\boldsymbol{\theta}^{\ast}\neq\boldsymbol{\theta}_{0}$. Then, we have
\[
\sqrt{n}(s(\widehat{\boldsymbol{\theta}}_{\tau})-s(\boldsymbol{\theta}%
^{\ast}))\overset{\mathcal{L}}{\underset{n\mathcal{\rightarrow}\infty
}{\longrightarrow}}\mathcal{N}\left(  0,\sigma_{W_{n}^0(\boldsymbol{\theta}%
_{0})}^{2}(\boldsymbol{\theta}^{\ast}))\right)  ,
\]
where
$
s\left(  \boldsymbol{\theta}\right)  =\left(  \boldsymbol{\theta
}-\boldsymbol{\theta}_{0}\right)  ^{T}\boldsymbol{\Sigma}_{\tau}%
^{-1}(\boldsymbol{\theta}_{0})\left(  \boldsymbol{\theta}-\boldsymbol{\theta
}_{0}\right)  
$
and
$
\sigma_{W_{n}^0(\boldsymbol{\theta}_{0})}^{2}(\boldsymbol{\theta}^{\ast
})=4\left(  \boldsymbol{\theta}^{\ast}-\boldsymbol{\theta}_0\right)^{T}
\left[\boldsymbol{\Sigma}_{\tau}^{-1}(\boldsymbol{\theta}_{0})
\boldsymbol{\Sigma}_{\tau}(\boldsymbol{\theta}^{\ast})
\boldsymbol{\Sigma}_{\tau}^{-1}(\boldsymbol{\theta}_{0})\right]
\left(\boldsymbol{\theta}^{\ast}-\boldsymbol{\theta}_0\right)  .
$
\label{THM:Papprox}
\end{theorem}

\begin{proof}
A first-order Taylor expansion of $s\left(  \boldsymbol{\theta}\right)  $
around $\boldsymbol{\theta}^{\ast}$ at $\widehat{\boldsymbol{\theta}}_{\tau
}$ is given by%
\[
s(\widehat{\boldsymbol{\theta}}_{\tau})-s(\boldsymbol{\theta}^{\ast
})=\left.  \frac{\partial s\left(  \boldsymbol{\theta}\right)  }%
{\partial\boldsymbol{\theta}^{T}}\right\vert _{\boldsymbol{\theta=\theta
}^{\ast}}(\widehat{\boldsymbol{\theta}}_{\tau}-\boldsymbol{\theta}^{\ast
})+o_{p}(n^{-1/2}).
\]
Then the asymptotic distribution of the random variable $\sqrt{n}%
(s(\widehat{\boldsymbol{\theta}}_{\tau})-s(\boldsymbol{\theta}^{\ast}))$
matches the asymptotic distribution of the random variable $\left.
\frac{\partial s\left(  \boldsymbol{\theta}\right)  }{\partial
\boldsymbol{\theta}^{T}}\right\vert _{\boldsymbol{\theta=\theta}^{\ast}%
}\sqrt{n}\left(\widehat{\boldsymbol{\theta}}_{\tau}-\boldsymbol{\theta}^{\ast}\right)$ 
and the desired result follows.
\end{proof}

Based on Theorem \ref{THM:null} we shall reject the null hypothesis given in (\ref{a.6}) if
\begin{equation}
W_{n}^0(\boldsymbol{\theta}_{0}) > \chi_{p,\alpha}^{2},
\label{EQ:CR}
\end{equation}
and Theorem \ref{THM:Papprox} makes it possible to have an approximation 
of the power function for the test given in (\ref{EQ:CR}). This is given by 
%Based on this result we have an approximation of the power function, at the level of significance $\alpha$, as given by%
\begin{align}
\pi^{\tau}_{W_{n}^0(\boldsymbol{\theta}_{0})}(\boldsymbol{\theta}^{\ast}) &
=\Pr\left(  \text{Rejecting }H_{0}| \boldsymbol{\theta}=\boldsymbol{\theta}^\ast\right)  
=\Pr\left(  W_{n}^0(\boldsymbol{\theta}_{0})>\chi_{p,\alpha}^{2}|\boldsymbol{\theta}=\boldsymbol{\theta}^{\ast}\right)  
\nonumber\\
&  =\Pr\left(  s(\widehat{\boldsymbol{\theta}}_{\tau})-s(\boldsymbol{\theta
}^{\ast})>\frac{\chi_{p,\alpha}^{2}}{n}-s(\boldsymbol{\theta}^{\ast})\right)
\nonumber\\
&  =1-\Phi_{n}\left(  \frac{n^{1/2}}{\sigma_{W_{n}^0(\boldsymbol{\theta}_{0}%
)}(\boldsymbol{\theta}^{\ast})}\left(  \frac{\chi_{p,\alpha}^{2}}%
{n}-s(\boldsymbol{\theta}^{\ast})\right)  \right)  ,\label{EQ:power_approx}
\end{align}
where $\chi_{p,\alpha}^2$ denote the $(1-\alpha)$-th quantile of $\chi_p^2$ distribution,
and $\Phi_{n}(\cdot)$ is a sequence of distribution functions tending uniformly
to the standard normal distribution function $\Phi(\cdot)$. We can observe
that the Wald-type tests are consistent in the Fraser sense since%
\begin{eqnarray}
\lim_{n\rightarrow\infty}\pi^{\tau}_{W_{n}^0(\boldsymbol{\theta}_{0})}%
(\boldsymbol{\theta}^{\ast})=1\text{ }\forall\tau\geq 0.
\label{EQ:power}
\end{eqnarray}
This result can be applied in the sense of getting the necessary sample size
for the Wald-type tests to have a predetermined power, 
$\pi^{\tau}_{W_{n}^0(\boldsymbol{\theta}_{0})}(\boldsymbol{\theta}^{\ast})\approx\pi^{\ast}$ 
and size $\alpha$. The necessary sample size is given by
\[
n=\left[  \frac{A+B+\sqrt{A(A+2B)}}{2s^{2}(\boldsymbol{\theta}^{\ast}%
)}\right]  +1
\]
where $\left[  z\right]  $ denotes the largest integer less than or equal to $z$, 
$A   =\sigma_{W_{n}^0(\boldsymbol{\theta}_{0})}^{2}(\boldsymbol{\theta}^{\ast})
\left(  \Phi^{-1}(1-\pi^{\ast})\right)  ^{2},$
and
$B  =2s(\boldsymbol{\theta}^{\ast})\chi_{p,\alpha}^{2}$.
%\begin{align*}
%A &  =\sigma_{W_{n}^0(\boldsymbol{\theta}_{0})}^{2}(\boldsymbol{\theta}^{\ast
%})\left(  \Phi^{-1}(1-\pi^{\ast})\right)  ^{2},\\
%B &  =2s(\boldsymbol{\theta}^{\ast})\chi_{p,\tau}^{2}.
%\end{align*}

In order to produce a non-trivial asymptotic power, see (\ref{EQ:power}), 
\cite{Cochran:1952} suggested using a set of local alternatives contiguous to the null hypothesis given in (\ref{a.6})
as $n$ increases. In the next theorem we shall
present the asymptotic distribution of the Wald-type test statistics under
contiguous alternative hypotheses.

\begin{theorem}
Under the contiguous alternative hypothesis%
\begin{equation}
H_{1,n}:\boldsymbol{\theta}_{n}=\boldsymbol{\theta}_{0}+n^{-1/2}%
\boldsymbol{d}, \label{a.8}%
\end{equation}
where $\boldsymbol{d}$ is a fixed vector in $\mathbb{R}^{p}$ such that
$\boldsymbol{\theta}_{n}\in\Theta\subset\mathbb{R}^{p}$, the asymptotic
distribution of the Wald-type test statistics given in (\ref{a.7}) is $\chi_p^2(\delta)$,
a non-central chi-square distribution with $p$ degrees of freedom and
non-centrality parameter%
\begin{equation}
\delta=\boldsymbol{d}^{T}\boldsymbol{\Sigma}_{\tau}^{-1}(\boldsymbol{\theta
}_{0})\boldsymbol{d}. \label{a.9}%
\end{equation}

\end{theorem}

\begin{proof}
We can write
\[
\widehat{\boldsymbol{\theta}}_{\tau}-\boldsymbol{\theta}_0=\widehat{\boldsymbol{\theta}}_{\tau}
-\boldsymbol{\theta}_{n}+\boldsymbol{\theta}_{n}-\boldsymbol{\theta}_0
=(\widehat{\boldsymbol{\theta}}_{\tau}-\boldsymbol{\theta}_{n})+n^{-1/2}\boldsymbol{d}.
\]
Therefore, under $H_{1,n}$ given in (\ref{a.8}), we have from (\ref{a.25}),
\[
\sqrt{n}(\widehat{\boldsymbol{\theta}}_{\tau}-\boldsymbol{\theta}%
_{n})\overset{\mathcal{L}}{\underset{n\mathcal{\rightarrow}\infty
}{\longrightarrow}}\mathcal{N}\left(  \boldsymbol{0},\boldsymbol{\Sigma
}_{\tau}(\boldsymbol{\theta}_{0})\right)
\]
and hence
\[
\sqrt{n}(\widehat{\boldsymbol{\theta}}_{\tau}-\boldsymbol{\theta}%
_{0})\overset{\mathcal{L}}{\underset{n\mathcal{\rightarrow}\infty
}{\longrightarrow}}\mathcal{N}\left(  \boldsymbol{d},\boldsymbol{\Sigma
}_{\tau}(\boldsymbol{\theta}_{0})\right)  .
\]
Note that $W_{n}^0(\boldsymbol{\theta}_{0})$ can be written by
\[
W_{n}(\boldsymbol{\theta}_{0})=\left(  n^{1/2}\boldsymbol{\Sigma}_{\tau
}^{-1/2}(\boldsymbol{\theta}_{0})(\widehat{\boldsymbol{\theta}}_{\tau
}-\boldsymbol{\theta}_{0})\right)  ^{T}\left(  n^{1/2}\boldsymbol{\Sigma
}_{\tau}^{-1/2}(\boldsymbol{\theta}_{0})(\widehat{\boldsymbol{\theta}%
}_{\tau}-\boldsymbol{\theta}_{0})\right)
\]
and under $H_{1,n}$ given in (\ref{a.8}), we have,%
\[
n^{1/2}\boldsymbol{\Sigma}_{\tau}^{-1/2}(\boldsymbol{\theta}_{0}%
)(\widehat{\boldsymbol{\theta}}_{\tau}-\boldsymbol{\theta}_{0}%
)\overset{\mathcal{L}}{\underset{n\mathcal{\rightarrow}\infty}{\longrightarrow
}}\mathcal{N}\left(  \boldsymbol{\Sigma}_{\tau}^{-1/2}(\boldsymbol{\theta
}_{0})\boldsymbol{d},\boldsymbol{I}_{p\times p}\right)  .
\]
We apply the following result concerning quadratic forms. \textquotedblleft If
$\boldsymbol{Z\sim}\mathcal{N}\left(  \boldsymbol{\mu},\boldsymbol{\Sigma
}\right)  $, $\boldsymbol{\Sigma}$ is a symmetric projection of rank $k$ and
$\boldsymbol{\Sigma\mu=\mu}$, then $\boldsymbol{Z}^{T}\boldsymbol{Z}$ is a
chi-square distribution with $k$ degrees of freedom and non-centrality
parameter $\boldsymbol{\mu}^{T}\boldsymbol{\mu}$\textquotedblright. Therefore
\[
W_{n}^0(\boldsymbol{\theta}_{0})\overset{\mathcal{L}%
}{\underset{n\mathcal{\rightarrow}\infty}{\longrightarrow}}\chi_{p}^{2}\left(
\delta\right)  ,
\]
where $\delta$ was defined in (\ref{a.9}).
\end{proof}

The last theorem permits us to get an approximation to the power function at
$\boldsymbol{\theta}_{n}$ by
\begin{eqnarray}
\pi^{\tau}_{W_{n}^0(\boldsymbol{\theta}_{0})}(\boldsymbol{\theta}_{n})=1-G_{\chi
_{p}^{2}\left(  \delta\right)  }\left(  \chi_{p,\alpha}^{2}\right)  ,
\label{EQ:Power_one}
\end{eqnarray}
where $G_{\chi_{p}^{2}\left(  \delta\right)  }\left(  z\right)$ is the distribution function of $\chi_{p}^{2}\left(  \delta\right) $
%a non-central chi-square, with $p$ degrees of freedom and non-centrality parameter $\delta=\boldsymbol{d}%
%^{T}\boldsymbol{\Sigma}_{\tau}^{-1}(\boldsymbol{\theta}_{0})\boldsymbol{d}$
evaluated at the point $z$.

Based on this result we can also obtain an approximation of the power function at a
generic point $\boldsymbol{\theta}^{\ast}$, because we can consider
$\boldsymbol{d}=n^{1/2}\left(  \boldsymbol{\theta}^{\ast}-\boldsymbol{\theta
}_{0}\right)  $ and then $\boldsymbol{\theta}_{n}=\boldsymbol{\theta}^{\ast}$.

\subsection{Wald-type tests for composite null hypotheses}\label{sec2.2}

In many practical hypothesis testing problems, the restricted parameter space
$\Theta_{0}\subset\Theta$ is defined by a set of $r<p$ non-redundant
restrictions of the form
\begin{equation}
\boldsymbol{h}(\boldsymbol{\theta})=\boldsymbol{0} \label{a.10}%
\end{equation}
on $\Theta$, where $\boldsymbol{h}:\mathbb{R}^{p}\rightarrow\mathbb{R}^{r}$ is
a vector-valued function such that the full rank $p\times r$ matrix
\begin{equation}
\boldsymbol{H}\left(  \boldsymbol{\theta}\right)  =\frac{\partial
\boldsymbol{h}(\boldsymbol{\theta})}{\partial\boldsymbol{\theta}^{T}}
\label{a.11}%
\end{equation}
exists and is continuous in $\boldsymbol{\theta}$. 
%For more details see Sen and Singer (1993).

Our interest will be in testing
\begin{equation}
H_{0}:\boldsymbol{\theta}\in\Theta_{0}\subset{\mathbb{R}}^{p-r}\text{ against
}H_{1}:\boldsymbol{\theta}\in\Theta-\Theta_{0} \label{a.12}%
\end{equation}
using Wald-type test statistics.

\begin{definition}
Let $\widehat{\boldsymbol{\theta}}_{\tau}$ be the MDPDE of
$\boldsymbol{\theta}$. The family of proposed Wald-type test statistics for
testing the composite null hypothesis (\ref{a.12}) is given by%
\begin{equation}
W_{n}(\widehat{\boldsymbol{\theta}}_{\tau})=n\boldsymbol{h}^{T}%
(\widehat{\boldsymbol{\theta}}_{\tau})\left(  \boldsymbol{H}^{T}%
(\widehat{\boldsymbol{\theta}}_{\tau})\boldsymbol{\Sigma}_{\tau
}(\widehat{\boldsymbol{\theta}}_{\tau})\boldsymbol{H}%
(\widehat{\boldsymbol{\theta}}_{\tau})\right)  ^{-1}\boldsymbol{h}%
(\widehat{\boldsymbol{\theta}}_{\tau}), \label{a.13}%
\end{equation}

\end{definition}

In the next theorem we are going to present the asymptotic distribution of
$W_{n}(\widehat{\boldsymbol{\theta}}_{\tau})$.

\begin{theorem}
The asymptotic distribution of the Wald-type test statistics $W_{n}%
(\widehat{\boldsymbol{\theta}}_{\tau})$, under the null hypothesis given in
(\ref{a.12}), is chi-squared with $r$ degrees of freedom.
\label{THM:null_composite}
\end{theorem}

\begin{proof}
Let $\boldsymbol{\theta}_{0}\in\Theta_{0}$ be the true value of the parameter. A
Taylor expansion gives
\begin{align*}
\boldsymbol{h}(\widehat{\boldsymbol{\theta}}_{\tau})  &  =\boldsymbol{h}%
(\boldsymbol{\theta}_{0})\boldsymbol{+H}^{T}(\boldsymbol{\theta}%
_{0})(\widehat{\boldsymbol{\theta}}_{\tau}-\boldsymbol{\theta}_{0}%
)+o_{p}(n^{-1/2}\boldsymbol{1})\\
&  =\boldsymbol{H}^{T}(\boldsymbol{\theta}_{0})(\widehat{\boldsymbol{\theta}%
}_{\tau}-\boldsymbol{\theta}_{0})+o_{p}(n^{-1/2}\boldsymbol{1}).
\end{align*}
Under $H_{0}$
\[
\sqrt{n}(\widehat{\boldsymbol{\theta}}_{\tau}-\boldsymbol{\theta}%
_{0})\overset{\mathcal{L}}{\underset{n\mathcal{\rightarrow}\infty
}{\longrightarrow}}\mathcal{N}\left(  \boldsymbol{0}_{p},\boldsymbol{\Sigma
}_{\tau}(\boldsymbol{\theta}_{0})\right)
\]
and so
\[
\sqrt{n}\boldsymbol{h}(\widehat{\boldsymbol{\theta}}_{\tau}%
)\overset{\mathcal{L}}{\underset{n\mathcal{\rightarrow}\infty}{\longrightarrow
}}\mathcal{N}\left(  \boldsymbol{0}_{p},\boldsymbol{H}^{T}(\boldsymbol{\theta
}_{0})\boldsymbol{\Sigma}_{\tau}(\boldsymbol{\theta}_{0})\boldsymbol{H}%
(\boldsymbol{\theta}_{0})\right)  .
\]
Taking into account that $\mathrm{rank}\left(  \boldsymbol{H}\left(
\boldsymbol{\theta}\right)  \right)  =r$, we get
\[
n\boldsymbol{h}^{T}(\widehat{\boldsymbol{\theta}}_{\tau})\left(
\boldsymbol{H}^{T}(\boldsymbol{\theta}_{0})\boldsymbol{\Sigma}_{\tau
}(\boldsymbol{\theta}_{0})\boldsymbol{H}(\boldsymbol{\theta}_{0})\right)
^{-1}\boldsymbol{h}(\widehat{\boldsymbol{\theta}}_{\tau}%
)\overset{\mathcal{L}}{\underset{n\mathcal{\rightarrow}\infty}{\longrightarrow
}}\chi_{r}^{2}.
\]
But
$
\boldsymbol{H}^{T}(\widehat{\boldsymbol{\theta}}_{\tau})\boldsymbol{\Sigma
}_{\tau}(\widehat{\boldsymbol{\theta}}_{\tau})\boldsymbol{H}%
(\widehat{\boldsymbol{\theta}}_{\tau})
$
is a consistent estimator of 
$\boldsymbol{H}^{T}(\boldsymbol{\theta}_{0})\boldsymbol{\Sigma}_{\tau}(\boldsymbol{\theta}_{0})\boldsymbol{H}(\boldsymbol{\theta}_{0})$
by continuity of the matrices $\boldsymbol{H}(\boldsymbol{\theta})$ and $\boldsymbol{\Sigma}_{\tau}(\boldsymbol{\theta})$
at $\boldsymbol{\theta}=\boldsymbol{\theta}_0$. 
Then, it holds that
\[
W_{n}(\widehat{\boldsymbol{\theta}}_{\tau})\overset{\mathcal{L}%
}{\underset{n\mathcal{\rightarrow}\infty}{\longrightarrow}}\chi_{r}^{2}.
\]

\end{proof}

Based on the previous result, we shall reject the null hypothesis given in
(\ref{a.12}) if
\begin{equation}
W_{n}(\widehat{\boldsymbol{\theta}}_{\tau})>\chi_{r,\alpha}^{2}.
\label{a.14}%
\end{equation}
It is not easy to get an exact expression for the power function of the test
given in (\ref{a.14}). For that reason we are going to present a theorem that
will be important in order to get an approximation of the power function for
the test statistic presented in (\ref{a.14}).

\begin{theorem}
Let $\boldsymbol{\theta}^{\ast}\notin\Theta_{0}$ the true value of the
parameter with $\widehat{\boldsymbol{\theta}}_{\tau}\overset{\mathcal{P}%
}{\underset{n\mathcal{\rightarrow}\infty}{\longrightarrow}}\boldsymbol{\theta
}^{\ast}$. Define 
$$
s^{\ast}\left(  \boldsymbol{\theta}_{1},\boldsymbol{\theta}_{2}\right)
=\boldsymbol{h}^{T}(\boldsymbol{\theta}_{1}\boldsymbol{)}\left(
\boldsymbol{H}^{T}(\boldsymbol{\theta}_{2})\boldsymbol{\Sigma}_{\tau
}(\boldsymbol{\theta}_{2})\boldsymbol{H}(\boldsymbol{\theta}_{2})\right)
^{-1}\boldsymbol{h}(\boldsymbol{\theta}_{1}).
$$
Then, we have
\[
\sqrt{n}\left(  s^{\ast}(\widehat{\boldsymbol{\theta}}_{\tau}%
,\widehat{\boldsymbol{\theta}}_{\tau})-s^{\ast}(\boldsymbol{\theta}^{\ast
},\boldsymbol{\theta}^{\ast})\right)  \overset{\mathcal{L}%
}{\underset{n\mathcal{\rightarrow}\infty}{\longrightarrow}}\mathcal{N}\left(
0,\sigma^{2}(\boldsymbol{\theta}^{\ast})\right)  ,
\]
where
\[
\sigma^{2}(\boldsymbol{\theta}^{\ast})=4\boldsymbol{h}^{T}(\boldsymbol{\theta}^\ast)
\left(\boldsymbol{H}^{T}(\boldsymbol{\theta}^\ast)\boldsymbol{\Sigma}_{\tau}(\boldsymbol{\theta}^\ast)
\boldsymbol{H}(\boldsymbol{\theta}^\ast)\right)^{-1}\boldsymbol{h}(\boldsymbol{\theta}^\ast).
\]

\end{theorem}

\begin{proof}
We can observe that $s^{\ast}(\widehat{\boldsymbol{\theta}}_{\tau
},\widehat{\boldsymbol{\theta}}_{\tau})$ and $s^{\ast}%
(\widehat{\boldsymbol{\theta}}_{\tau},\boldsymbol{\theta}^{\ast})$ have the
same asymptotic distribution because $\widehat{\boldsymbol{\theta}}_{\tau
}\overset{\mathcal{P}}{\underset{n\mathcal{\rightarrow}\infty}{\longrightarrow
}}\boldsymbol{\theta}^{\ast}.$ A first-order Talyor expansion of 
$s^{\ast}({\boldsymbol{\theta}},\boldsymbol{\theta}^{\ast})$ at
$\widehat{\boldsymbol{\theta}}_{\tau}$ around $\boldsymbol{\theta}^{\ast}$
gives%
\[
s^{\ast}(\widehat{\boldsymbol{\theta}}_{\tau},\boldsymbol{\theta}^{\ast})
-s^{\ast}(\boldsymbol{\theta}^{\ast},\boldsymbol{\theta}^{\ast})
=\left.  \frac{\partial s^{\ast}\left(  \boldsymbol{\theta}%
,\boldsymbol{\theta}^{\ast}\right)  }{\partial\boldsymbol{\theta}^{T}%
}\right\vert _{\boldsymbol{\theta=\theta}^{\ast}}(\widehat{\boldsymbol{\theta
}}_{\tau}-\boldsymbol{\theta}^{\ast})+o_{p}(n^{-1/2}).
\]
Now the result follows, because
\[
\sigma^{2}(\boldsymbol{\theta}^{\ast})=\left.  \frac{\partial s^{\ast}\left(
	\boldsymbol{\theta},\boldsymbol{\theta}^{\ast}\right)  }{\partial
	\boldsymbol{\theta}^{T}}\right\vert _{\boldsymbol{\theta=\theta}^{\ast}%
}\boldsymbol{\Sigma}_{\tau}(\boldsymbol{\theta}^{\ast})\left.
\frac{\partial s^{\ast}\left(  \boldsymbol{\theta},\boldsymbol{\theta}^{\ast
	}\right)  }{\partial\boldsymbol{\theta}}\right\vert _{\boldsymbol{\theta
		=\theta}^{\ast}}.
\]
and 
$$ \frac{\partial s^{\ast}\left(
	\boldsymbol{\theta},\boldsymbol{\theta}^{\ast}\right)  }{\partial\boldsymbol{\theta}^{T}}
=2\boldsymbol{h}^{T}(\boldsymbol{\theta})
\left(\boldsymbol{H}^{T}(\boldsymbol{\theta}^\ast)\boldsymbol{\Sigma}_{\tau}(\boldsymbol{\theta}^\ast)
\boldsymbol{H}(\boldsymbol{\theta}^\ast)\right)^{-1}\boldsymbol{H}^T(\boldsymbol{\theta}).
$$
\end{proof}

Using the above theorem we can derive an approximation to the power 
of the proposed Wald-type tests of composite null hypothesis at any $\boldsymbol{\theta}^{\ast}\notin\Theta_{0}$
using an argument similar to that of the derivation of the expression in (\ref{EQ:power_approx})  for the case of simple null hypothesis. 
This further indicates the consistency of our proposal at any fixed alternatives even for the composite hypotheses.

We may also find an approximation of the power of $W_{n}(\widehat{\boldsymbol{\theta}}_{\tau})$ at an alternative
close to the null hypothesis. Let $\boldsymbol{\theta}_{n}\in\Theta-\Theta
_{0}$ be a given alternative and let $\boldsymbol{\theta}_{0}$ be the element
in $\Theta_{0}$ closest to $\boldsymbol{\theta}_{n}$ in the Euclidean distance
sense. A first possibility to introduce contiguous alternative hypotheses is
to consider a fixed $\boldsymbol{d}\in\mathbb{R}^{p}$ and to permit
$\boldsymbol{\theta}_{n}$ to move towards $\boldsymbol{\theta}_{0}$ as $n$
increases through the relation
\begin{equation}
H_{1,n}:\boldsymbol{\theta}_{n}=\boldsymbol{\theta}_{0}+n^{-1/2}%
\boldsymbol{d}. \label{a.15}%
\end{equation}
A second approach is to relax the condition $\boldsymbol{h}\left(
\boldsymbol{\theta}\right)  =\boldsymbol{0}$ defining $\Theta_{0}.$ Let
$\boldsymbol{d}^{\ast}\in\mathbb{R}^{r}$ and consider the following sequence,
$\left\{  \boldsymbol{\theta}_{n}\right\}  $, of parameters moving towards
$\boldsymbol{\theta}_{0}$ according to
\begin{equation}
H_{1,n}^{\ast}:\boldsymbol{h}(\boldsymbol{\theta}_{n})=n^{-1/2}\boldsymbol{d}%
^{\ast}. \label{a.16}%
\end{equation}
Note that a Taylor series expansion of $\boldsymbol{h}(\boldsymbol{\theta}%
_{n})$ around $\boldsymbol{\theta}_{0}$ yields
\begin{eqnarray}
\boldsymbol{h}(\boldsymbol{\theta}_{n})=\boldsymbol{h}(\boldsymbol{\theta}%
_{0})+\boldsymbol{H}^{T}(\boldsymbol{\theta}_{0})\left(  \boldsymbol{\theta
}_{n}-\boldsymbol{\theta}_{0}\right)  +o\left(  \left\Vert \boldsymbol{\theta
}_{n}-\boldsymbol{\theta}_{0}\right\Vert \boldsymbol{1}\right).
\label{EQ:b}
\end{eqnarray}
By substituting $\boldsymbol{\theta}_{n}=\boldsymbol{\theta}_{0}%
+n^{-1/2}\boldsymbol{d}$ in (\ref{EQ:b}) and taking into account
that\textbf{\ }$\boldsymbol{h}(\boldsymbol{\theta}_{0})=\boldsymbol{0}$, we
get
\begin{equation}
\boldsymbol{h}(\boldsymbol{\theta}_{n})=n^{-1/2}\boldsymbol{H}^{T}%
(\boldsymbol{\theta}_{0})\boldsymbol{d}+o\left(  \left\Vert \boldsymbol{\theta
}_{n}-\boldsymbol{\theta}_{0}\right\Vert \boldsymbol{1} \right)  , \label{a.17}%
\end{equation}
so that the equivalence in the limit is obtained for $\boldsymbol{d}^{\ast
}\boldsymbol{=H}^{T}(\boldsymbol{\theta}_{0})\boldsymbol{d}$.

\begin{theorem}
%We have the following results under both versions of the contiguous alternative hypothesis:
Under the contiguous alternative hypotheses given in (\ref{a.15}) and (\ref{a.16}), we have

\begin{enumerate}
\item[i)] $W_{n}(\widehat{\boldsymbol{\theta}}_{\tau})
\underset{n\rightarrow\infty}{\overset{\mathcal{L}}{\longrightarrow}}\chi_{r}^{2}\left(a\right)$ under
$H_{1,n}$ given in (\ref{a.15}), where the non-centrality parameter ``$a$" is given by
$$
a = \boldsymbol{d}^{T}\boldsymbol{H}(\boldsymbol{\theta}_{0})\left(
\boldsymbol{H}^{T}(\boldsymbol{\theta}_{0})\boldsymbol{\Sigma}_{\tau
}(\boldsymbol{\theta}_{0})\boldsymbol{H}(\boldsymbol{\theta}_{0})\right)
^{-1}\boldsymbol{H}^{T}(\boldsymbol{\theta}_{0})\boldsymbol{d}.
$$

\item[ii)] $W_{n}(\widehat{\boldsymbol{\theta}}_{\tau})
\underset{n\rightarrow\infty}{\overset{\mathcal{L}}{\longrightarrow}}\chi_{r}^{2}\left(b\right)  $ under $H_{1,n}^{\ast}$ given in
(\ref{a.16}), where the non-centrality parameter ``$b$" is given by
$$
b =   \boldsymbol{\boldsymbol{d}}^{\ast T}\left(  \boldsymbol{H}%
^{T}(\boldsymbol{\theta}_{0})\boldsymbol{\Sigma}_{\tau}(\boldsymbol{\theta
}_{0})\boldsymbol{H}(\boldsymbol{\theta}_{0})\right)  ^{-1}%
\boldsymbol{\boldsymbol{d}}^{\ast}.
$$
\end{enumerate}

\begin{proof}
A Taylor series expansion of $\boldsymbol{h}(\widehat{\boldsymbol{\theta}%
}_{{\tau}})$ around $\boldsymbol{\theta}_{n}$ yields%
\[
\boldsymbol{h}(\widehat{\boldsymbol{\theta}}_{{\tau}})=\boldsymbol{h}%
(\boldsymbol{\theta}_{n})+\boldsymbol{H}^{T}(\boldsymbol{\theta}%
_{n})(\widehat{\boldsymbol{\theta}}_{{\tau}}-\boldsymbol{\theta}_{n})+o\left(
\left\Vert \widehat{\boldsymbol{\theta}}_{{\tau}}-\boldsymbol{\theta}%
_{n}\right\Vert \boldsymbol{1}\right)  .
\]
From (\ref{a.17}), we have
\[
\boldsymbol{h}(\widehat{\boldsymbol{\theta}}_{{\tau}})=n^{-1/2}\boldsymbol{H}%
^{T}(\boldsymbol{\theta}_{0})\boldsymbol{d}+\boldsymbol{H}^{T}%
(\boldsymbol{\theta}_{n})(\widehat{\boldsymbol{\theta}}_{{\tau}}%
-\boldsymbol{\theta}_{n})+o\left(  \left\Vert \widehat{\boldsymbol{\theta}%
}_{{\tau}}-\boldsymbol{\theta}_{n}\right\Vert \boldsymbol{1}\right)  +o\left(
\left\Vert \boldsymbol{\theta}_{n}-\boldsymbol{\theta}_{0}\right\Vert
\boldsymbol{1}\right)  .
\]
As
$
\sqrt{n}(\widehat{\boldsymbol{\theta}}_{{\tau}}-\boldsymbol{\theta}%
_{n})\underset{n\rightarrow\infty}{\overset{\mathcal{L}}{\longrightarrow}%
}\mathcal{N}(\boldsymbol{0},\boldsymbol{\Sigma}_{\tau}(\boldsymbol{\theta
}_{0}))
$
and $\sqrt{n}\left(  o\left(  \left\Vert \widehat{\boldsymbol{\theta}}_{{\tau}
}-\boldsymbol{\theta}_{n}\right\Vert \boldsymbol{1}\right)  +o\left(  \left\Vert
\boldsymbol{\theta}_{n}-\boldsymbol{\theta}_{0}\right\Vert \boldsymbol{1}\right)  \right)
=o_{p}\left( \boldsymbol{1}\right)  $, we have
\[
\sqrt{n}\boldsymbol{h}(\widehat{\boldsymbol{\theta}}_{{\tau}}%
)\underset{n\rightarrow\infty}{\overset{\mathcal{L}}{\longrightarrow}%
}\mathcal{N}(\boldsymbol{H}^{T}(\boldsymbol{\theta}_{0})\boldsymbol{d}%
,\boldsymbol{H}^{T}(\boldsymbol{\theta}_{0})\boldsymbol{\Sigma}_{\tau
}(\boldsymbol{\theta}_{0})\boldsymbol{H}(\boldsymbol{\theta}_{0})).
\]
We can observe by the relationship $\boldsymbol{d}^{\ast}\boldsymbol{=H}%
^{T}(\boldsymbol{\theta}_{0})\boldsymbol{d}$, if $\boldsymbol{h}%
(\boldsymbol{\theta}_{n})=n^{-1/2}\boldsymbol{d}^{\ast}$ that
\[
\sqrt{n}\boldsymbol{h}(\widehat{\boldsymbol{\theta}}_{{\tau}}%
)\underset{n\rightarrow\infty}{\overset{\mathcal{L}}{\longrightarrow}%
}\mathcal{N}(\boldsymbol{d}^{\ast},\boldsymbol{H}^{T}(\boldsymbol{\theta}%
_{0})\boldsymbol{\Sigma}_{\tau}(\boldsymbol{\theta}_{0})\boldsymbol{H}%
(\boldsymbol{\theta}_{0})).
\]
In our case, the quadratic form is
$W_{n}=\boldsymbol{Z}^{T}\boldsymbol{Z}$
with
$\boldsymbol{Z}=\sqrt{n}\boldsymbol{h}(\widehat{\boldsymbol{\theta}}_{{\tau}
})\left(  \boldsymbol{H}^{T}(\boldsymbol{\theta}_{0})\boldsymbol{\Sigma
}_{\tau}(\boldsymbol{\theta}_{0})\boldsymbol{H}(\boldsymbol{\theta}%
_{0})\right)  ^{-1/2}$
and
\[
\boldsymbol{Z}\underset{n\rightarrow\infty}{\overset{\mathcal{L}%
}{\longrightarrow}}\mathcal{N}\left(  \left(  \boldsymbol{H}^{T}%
(\boldsymbol{\theta}_{0})\boldsymbol{\Sigma}_{\tau}(\boldsymbol{\theta}%
_{0})\boldsymbol{H}(\boldsymbol{\theta}_{0})\right)  ^{-1/2}\boldsymbol{H}%
^{T}(\boldsymbol{\theta}_{0})\boldsymbol{d},\boldsymbol{I}\right)  ,
\]
where $\boldsymbol{I}$ is the identity $r\times r$ matrix. Hence, the
application of the result is immediate and the non-centrality parameter is
\[
\boldsymbol{d}^{T}\boldsymbol{H}(\boldsymbol{\theta}_{0})\left(
\boldsymbol{H}^{T}(\boldsymbol{\theta}_{0})\boldsymbol{\Sigma}_{\tau
}(\boldsymbol{\theta}_{0})\boldsymbol{H}(\boldsymbol{\theta}_{0})\right)
^{-1}\boldsymbol{H}^{T}(\boldsymbol{\theta}_{0})\boldsymbol{d}=\boldsymbol{d}%
^{\ast T}\left(  \boldsymbol{H}^{T}(\boldsymbol{\theta}_{0})\boldsymbol{\Sigma
}_{\tau}(\boldsymbol{\theta}_{0})\boldsymbol{H}(\boldsymbol{\theta}%
_{0})\right)  ^{-1}\boldsymbol{d}^{\ast}.
\]
\end{proof}
\end{theorem}

\section{Robustness of the Wald-type tests for Non-homogeneous Observations\label{sec3}}

\subsection{Influence functions of the Wald-type test statistics}
\label{sec3.1}

In order to study the robustness of a testing procedure, 
the first measure to consider is Hampel's influence function (IF) of the test statistics,
introduced by \cite{Rousseeuw/Ronchetti:1979} for i.i.d.~data; 
see also \cite{Rousseeuw/Ronchetti:1981} and \cite{Hampel/etc:1986} for detail.
In case of non-homogeneous data, the concept of IF has been extended suitably  
by \cite{Huber:1983} and \cite{Ghosh/Basu:2013, Ghosh/Basu:2016} for the estimators
and by \cite{Ghosh/Basu:2015} and \cite{Aerts/Haesbroeck:2016} for test statistics. 
Here, we will follow these extended definitions of IF to study the robustness
of our proposed Wald-type test statistics for non-homogeneous observations.

In order to define and study the IF for the Wald-type test statistics, 
we first need the same for the MDPDE used in constructing the Wald-type test statistics;
we will briefly recall the IF of the MDPDE under non-homogeneous observations for the sake of completeness.
Suppose $G_i$ denote the true distribution of $Y_i$ having corresponding density $g_i$ for each $i=1,\ldots, n$;
under the model null distribution with true parameter value $\boldsymbol{\theta}_0$ we have
$G_i=F_{i,\boldsymbol{\theta}}$ and $g_i=f_{i,\boldsymbol{\theta}}$ for each $i$.
Denote $\underline{\boldsymbol{G}}= (G_1,\cdots,G_n)$ and
$\mathbf{\underline{F}}_{{\boldsymbol{\theta}}_0}=(F_{1,\boldsymbol{\theta}_0},\cdots, F_{n,\boldsymbol{\theta}_0})$.
Then the minimum DPD functional $
%{\boldsymbol{\theta}}^g = 
\boldsymbol{T}_{\tau}(\underline{\boldsymbol{G}})$ for independent but non-homogeneous observations 
at the true distribution $\underline{\boldsymbol{G}}$ 
is defined as the minimizer, with respect to ${\boldsymbol{\theta}} \in \Theta$, of 
the average DPD measure 
$\frac{1}{n} \displaystyle\sum_{i=1}^n d_{\tau}(g_i,f_{i,\boldsymbol{\theta}})$ with 
%$d_{\tau}(f_1, f_2)$ being the DPD between two densities $f_1$ and $f_2$,
%or equivalently, of the objective function 
%$\frac{1}{n} \sum_{i=1}^n H^{(i)}({\boldsymbol{\theta}})$
%with
%\begin{equation}\label{EQ:8Hi}
%H^{(i)}({\boldsymbol{\theta}})= \int f_{i, {\boldsymbol{\theta}}}(y)^{1+\tau} dy - 
%\left(1+\frac{1}{\tau}\right) \int f_{i, {\boldsymbol{\theta}}}(y)^{\tau} g_i(y) dy.
%\end{equation}
\begin{equation}
d_{\tau}(g_i,f_{i,\boldsymbol{\theta}}) = \int \left\{ f_{i,\boldsymbol{\theta}}(y)^{1+\tau} 
- \left(1+\frac{1}{\tau}\right)f_{i,\boldsymbol{\theta}}(y)^{\tau}g_i(y) 
+ \frac{1}{\tau}g_i(y)^{1+\tau}\right\}dy.
\nonumber
\end{equation}
Now, for each $i=1, \ldots, n$, let us denote by  $G_{i,\epsilon} = (1-\epsilon) G_i + \epsilon \wedge_{t_i}$
the $\epsilon$-contaminated distribution in the $i$-th direction, 
where $\wedge_{t_i}$ denotes the degenerate distribution at the contamination point $t_i$. 
%As in the estimation problem in \cite{Ghosh/Basu:2013}, 
Note that, in the case of non-homogeneous data the contamination can be either in any fixed direction,
say $i_0$-th direction, or in all the $n$ directions.
The corresponding IF of the minimum DPD functional $\boldsymbol{T}_{\tau}(\underline{\boldsymbol{G}})$
has been established in \cite{Ghosh/Basu:2013}; their forms at the model null distributions are given by
\begin{eqnarray}
IF_{i_0}(t_{i_0}; \boldsymbol{T}_{\tau}, \mathbf{\underline{F}}_{{\boldsymbol{\theta}}_0})
&=&  \boldsymbol\Psi_{n,\tau}^{-1}({\boldsymbol{\theta}}_0)
\frac{1}{n}\boldsymbol{D}_{\tau, i_0}(t_{i_0};{\boldsymbol{\theta}}_0),
\label{EQ:IF_MDPDE0}\\
IF(t_{1},\ldots,t_n; \boldsymbol{T}_{\tau}, \mathbf{\underline{F}}_{{\boldsymbol{\theta}}_0})
&=& \boldsymbol\Psi_{n,\tau}^{-1}({\boldsymbol{\theta}}_0)
\frac{1}{n} \sum_{i=1}^{n}\boldsymbol{D}_{\tau, i}(t_{i}; {\boldsymbol{\theta}}_0),
\label{EQ:IF_MDPDEa}
\end{eqnarray}
where $\boldsymbol{D}_{\tau, i}(t;{\boldsymbol{\theta}})
=\left[  f_{i,\boldsymbol{\theta}}(t)^{\tau} {\boldsymbol{u}_{i,\boldsymbol{\theta}}}(t) - \boldsymbol\xi_{i,\tau}\right]$
with $\boldsymbol\xi_{i,\tau}$ being as defined in Equation (\ref{EQ:xi}).
Note that these IFs are bounded at $\tau>0$ and unbounded at $\tau=0$,
implying the robustness of the MDPDEs with $\tau>0$ over the classical MLE (at $\tau=0$).

Now, we can defined the IF for the proposed Wald-type test statistics. 
We define the associated statistical functional, evaluated at $\underline{\boldsymbol{G}}$, as (ignoring the multiplier $n$)%
\begin{eqnarray}
W_{{\tau}}^{0}(\underline{\boldsymbol{G}}) &=&
(\boldsymbol{T}_{{\tau}}(\underline{\boldsymbol{G}})-\boldsymbol{\theta}_{0})^{T}
\boldsymbol{\Sigma}_{{\tau}}^{-1}(\boldsymbol{\theta}_{0})
(\boldsymbol{T}_{{\tau}}(\underline{\boldsymbol{G}})-\boldsymbol{\theta}_{0})
\label{EQ:WaldFunc0}%
\end{eqnarray}
corresponding to (\ref{a.7}) for the simple null hypothesis, and 
\begin{eqnarray}
W_{{\tau}}(\underline{\boldsymbol{G}}) &=&
\boldsymbol{h}^{T}(\boldsymbol{T}_{{\tau}}(\underline{\boldsymbol{G}}))
\left(  \boldsymbol{H}^{T}(\boldsymbol{T}_{{\tau}}(\underline{\boldsymbol{G}}))
\boldsymbol{\Sigma}_{\tau}(\boldsymbol{T}_{{\tau}}(\underline{\boldsymbol{G}}))
\boldsymbol{H}(\boldsymbol{T}_{{\tau}}(\underline{\boldsymbol{G}}))\right)^{-1}
\boldsymbol{h}(\boldsymbol{T}_{{\tau}}(\underline{\boldsymbol{G}})) 
\label{EQ:WaldFunc}%
\end{eqnarray}
corresponding to (\ref{a.13}) for the composite null hypothesis.

First we consider the Wald-type test functional $W_{{\tau}}^{0}$ for the simple null hypothesis
and contamination only one direction, say $i_0$-th direction.
The corresponding IF is then defined as  
\begin{eqnarray}
IF_{i_0}(t_{i_0}; W_{\tau}^{0}, \underline{\mathbf{G}})  
= \left.\frac{\partial}{\partial\epsilon} W_{\tau}^{0}(G_1, \cdots, G_{i_0-1}, G_{i_0,\epsilon},G_{i_0+1},\cdots, G_n)\right|_{\epsilon=0} 
%\nonumber\\	&=&  
= 2 (\boldsymbol{T}_{{\tau}}(\underline{\boldsymbol{G}})-\boldsymbol{\theta}_{0})^{T}
\boldsymbol{\Sigma}_{{\tau}}^{-1}(\boldsymbol{\theta}_{0})
IF_{i_0}(t_{i_0}; \boldsymbol{T}_{\tau}, \underline{\mathbf{G}}),\nonumber
\end{eqnarray}
which, when evaluated at the null distribution $\underline{\mathbf{G}}=\mathbf{\underline{F}}_{{\boldsymbol{\theta}}_0}$,
becomes identically zero as $\boldsymbol{T}_{\tau}(\mathbf{\underline{F}}_{{\boldsymbol{\theta}}_0}) = {\boldsymbol{\theta}}_0$.
So, one need to consider the second order IF of the proposed Wald-type test functional $W_{\tau}^{0}$ defined as
\begin{eqnarray}
IF_{i_0}^{(2)}(t_{i_0}; W_{\tau}^{0}, \underline{\mathbf{G}}) 
&=& \frac{\partial^2}{\partial^2\epsilon} 
W_{\tau}^{0}(G_1,\cdots,G_{i_0-1},G_{i_0,\epsilon},G_{i_0+1},\cdots, G_n) \big|_{\epsilon=0}. \nonumber 
\end{eqnarray}
When evaluated at the null model distribution  $\underline{\mathbf{G}}=\mathbf{\underline{F}}_{{\boldsymbol{\theta}}_0}$, 
this second order IF has the simplified form 
\begin{eqnarray}
IF_{i_0}^{(2)}(t_{i_0}; W_{\tau}^{0}, \mathbf{\underline{F}}_{{\boldsymbol{\theta}}_0}) 
&=& 2 IF_{i_0}(t_{i_0}; \boldsymbol{T}_{\tau}, \mathbf{\underline{F}}_{{\boldsymbol{\theta}}_0})^T 
\boldsymbol{\Sigma}_{{\tau}}^{-1}(\boldsymbol{\theta}_{0})
IF_{i_0}(t_{i_0}; \boldsymbol{T}_{\tau}, \mathbf{\underline{F}}_{{\boldsymbol{\theta}}_0}) \nonumber\\
&=& 2 \left[\frac{1}{n}\boldsymbol{D}_{\tau, i_0}(t_{i_0};{\boldsymbol{\theta}}_0)\right]^T 
\left[\boldsymbol\Psi_{n,\tau}^{-1}({\boldsymbol{\theta}}_0)
\boldsymbol{\Sigma}_{{\tau}}^{-1}(\boldsymbol{\theta}_{0})
\boldsymbol\Psi_{n,\tau}^{-1}({\boldsymbol{\theta}}_0)\right]
\left[\frac{1}{n}\boldsymbol{D}_{\tau, i_0}(t_{i_0};{\boldsymbol{\theta}}_0)\right]. 
\label{EQ:IF2W00}
\end{eqnarray}
Similarly, we can derive the first and second order IF of $W_{\tau}^{0}$ 
for contamination in all directions at the point $\boldsymbol{t}=(t_1,\ldots, t_n)$
respectively defined as 
\begin{eqnarray}
IF(\boldsymbol{t}; W_{\tau}^{0}, \underline{\mathbf{G}}) 
= \frac{\partial}{\partial\epsilon} W_{\tau}^{0}(G_{1,\epsilon},\cdots, G_{n,\epsilon})\big|_{\epsilon=0}, 
%\nonumber \\
\mbox{ and }~
IF^{(2)}(\boldsymbol{t}; W_{\tau}^{0}, \underline{\mathbf{G}}) 
= \frac{\partial^2}{\partial^2\epsilon} 
W_{\tau}^{0}(G_{1,\epsilon},\cdots, G_{n,\epsilon}) \big|_{\epsilon=0}. \nonumber 
\end{eqnarray}
A direct calculation shows that, at the simple null model distribution 
$\underline{\mathbf{G}}=\mathbf{\underline{F}}_{{\boldsymbol{\theta}}_0}$, 
these IFs simplifies to 
\begin{eqnarray}
IF(\boldsymbol{t}; W_{\tau}^{0}, \mathbf{\underline{F}}_{{\boldsymbol{\theta}}_0}) 
&=& 0, \nonumber\\
IF^{(2)}(\boldsymbol{t}; W_{\tau}^{0}, \mathbf{\underline{F}}_{{\boldsymbol{\theta}}_0}) 
&=& 2 IF(\boldsymbol{t}; \boldsymbol{T}_{\tau}, \mathbf{\underline{F}}_{{\boldsymbol{\theta}}_0})^T 
\boldsymbol{\Sigma}_{{\tau}}^{-1}(\boldsymbol{\theta}_{0})
IF(\boldsymbol{t}; \boldsymbol{T}_{\tau}, \mathbf{\underline{F}}_{{\boldsymbol{\theta}}_0}) \nonumber\\
&=& 2 \left[\frac{1}{n}\sum_{i=1}^n\boldsymbol{D}_{\tau, i}(t_{i};{\boldsymbol{\theta}}_0)\right]^T 
\left[\boldsymbol\Psi_{n,\tau}^{-1}({\boldsymbol{\theta}}_0)
\boldsymbol{\Sigma}_{{\tau}}^{-1}(\boldsymbol{\theta}_{0})
\boldsymbol\Psi_{n,\tau}^{-1}({\boldsymbol{\theta}}_0)\right]
\left[\frac{1}{n}\sum_{i=1}^n\boldsymbol{D}_{\tau, i}(t_{i};{\boldsymbol{\theta}}_0)\right]. 
\label{EQ:IF2W0a}
\end{eqnarray}
Note that, both the second order IF in (\ref{EQ:IF2W00}) and (\ref{EQ:IF2W0a}) of the Wald-type test functional $W_{\tau}^0$
for testing simple null hypothesis under contamination in one or all directions are bounded,
whenever the corresponding MDPDE functional has bounded IF, i.e., for any $\tau>0$. 
This implies robustness of our proposed Wald-type tests for simple null hypothesis with $\tau>0$.

Next we can similarly derive the first and second order IFs of 
the proposed Wald-type tests functional $W_{\tau}$ in (\ref{EQ:WaldFunc}) for composite null hypotheses. 
For brevity, we will skip the details and present only the final results under composite null 
$\underline{\mathbf{G}}=\mathbf{\underline{F}}_{{\boldsymbol{\theta}}_0}$ 
with  $\boldsymbol{\theta}_0 \in \Theta$. In particular, the first order IF for contamination in 
either one or all directions are both identically zero, i.e., 
$$
IF_{i_0}(t_{i_0}; W_{\tau}^{0}, \mathbf{\underline{F}}_{{\boldsymbol{\theta}}_0})=0,~~~ 
IF(\boldsymbol{t}; W_{\tau}, \mathbf{\underline{F}}_{{\boldsymbol{\theta}}_0})=0, 
$$
and the corresponding second order IF has the form 
\begin{eqnarray}
IF_{i_0}^{(2)}(t_{i_0}; W_{\tau}, \mathbf{\underline{F}}_{{\boldsymbol{\theta}}_0}) 
&=& 2 IF_{i_0}(t_{i_0}; \boldsymbol{T}_{\tau}, \mathbf{\underline{F}}_{{\boldsymbol{\theta}}_0})^T 
\boldsymbol{H}(\boldsymbol{\theta}_{0})
\left[\boldsymbol{H}^{T}(\boldsymbol{\theta}_{0})\boldsymbol{\Sigma}_{\tau}(\boldsymbol{\theta}_{0})
\boldsymbol{H}(\boldsymbol{\theta}_{0})\right]^{-1}
\boldsymbol{H}^T(\boldsymbol{\theta}_{0})IF_{i_0}(t_{i_0}; \boldsymbol{T}_{\tau}, \mathbf{\underline{F}}_{{\boldsymbol{\theta}}_0}) \nonumber\\
%\label{EQ:IF2W0}\\
IF^{(2)}(\boldsymbol{t}; W_{\tau}, \mathbf{\underline{F}}_{{\boldsymbol{\theta}}_0}) 
&=& 2 IF(\boldsymbol{t}; \boldsymbol{T}_{\tau}, \mathbf{\underline{F}}_{{\boldsymbol{\theta}}_0})^T 
\boldsymbol{H}(\boldsymbol{\theta}_{0})
\left[\boldsymbol{H}^{T}(\boldsymbol{\theta}_{0})\boldsymbol{\Sigma}_{\tau}(\boldsymbol{\theta}_{0})
\boldsymbol{H}(\boldsymbol{\theta}_{0})\right]^{-1}
\boldsymbol{H}^T(\boldsymbol{\theta}_{0})
IF(\boldsymbol{t}; \boldsymbol{T}_{\tau}, \mathbf{\underline{F}}_{{\boldsymbol{\theta}}_0}).\nonumber
%\label{EQ:IF2Wa}
\end{eqnarray}
Again these second order IFs are bounded for any $\tau>0$ and unbounded at $\tau=0$
implying the robustness of our proposed Wald-type tests for composite hypothesis testing also.

\subsection{Level and Power Influence Functions\label{sec3.2}}

We will now study the robustness of the level and power of 
the proposed Wald-type tests through the corresponding 
influence functions for their asymptotic level and powers 
\citep{Hampel/etc:1986,Heritier/Ronchetti:1994,Toma/Broniatowski:2010,Ghosh/Basu:2015}.
Noting the consistency of these proposed Wald-type tests, we consider their asymptotic power 
under the contiguous alternatives in (\ref{a.8}) and (\ref{a.15}) respectively 
for the simple and composite hypotheses. 
Additionally, considering suitable contamination over these alternatives and the null hypothesis,
we define the contaminated distributions, for each $i=1, \ldots, n$,  
$$
F_{i,n,\epsilon,t_i}^L = \left(1-\frac{\epsilon}{\sqrt{n}}\right) 
F_{i,{\boldsymbol{\theta}}_0}+\frac{\epsilon}{\sqrt{n}} \wedge_{t_i},
%$$
\mbox{ and }
%$$
F_{i,n,\epsilon,t_i}^P = \left(1-\frac{\epsilon}{\sqrt{n}}\right) 
F_{i,{\boldsymbol{\theta}}_n}+\frac{\epsilon}{\sqrt{n}} \wedge_{t_i},
$$
respectively for the analysis of level and power stability.
Denote $\mathbf{t} = (t_1, \cdots, t_n)^T$, 
$\mathbf{\underline{F}}_{n,\epsilon,\mathbf{t}}^P = (F_{1,n,\epsilon,t_i}^P, \cdots, F_{n,n,\epsilon,t_i}^P)$ and 
$\mathbf{\underline{F}}_{n,\epsilon,\mathbf{t}}^L = (F_{1,n,\epsilon,t_i}^L, \cdots, F_{n,n,\epsilon,t_i}^L)$. 
Then the level influence function (LIF) and the power influence function (PIF) 
for the proposed Wald-type test statistics $W_{n}^0(\boldsymbol{\theta}_{0})$ 
for the simple null hypothesis (\ref{a.6}) are defined, 
assuming the nominal level of significance to be $\alpha$, as
\begin{eqnarray}
LIF(\mathbf{t}; W_{n}^0, \mathbf{\underline{F}}_{{\boldsymbol{\theta}}_0} ) 
&=& \lim_{n \rightarrow \infty} ~ \frac{\partial}{\partial \epsilon} 
P_{\mathbf{\underline{F}}_{n,\epsilon,\mathbf{t}}^L }( W_{n}^0(\boldsymbol{\theta}_{0}) >  \chi_{p,\alpha}^2) 
\big|_{\epsilon=0},\nonumber\\
%	$$
%	and the power influence function (PIF) is defined as
%	$$
PIF(\mathbf{t}; W_{n}^0, \mathbf{\underline{F}}_{{\boldsymbol{\theta}}_0} ) 
&=& \lim_{n \rightarrow \infty} ~ \frac{\partial}{\partial \epsilon} 
P_{\mathbf{\underline{F}}_{n,\epsilon,\mathbf{t}}^P }(  W_{n}^0(\boldsymbol{\theta}_{0}) >  \chi_{p,\alpha}^2) 
\big|_{\epsilon=0}.\nonumber
\end{eqnarray}
%	$$
Similarly the LIF and PIF of the proposed Wald-type test statistics 
$W_{n}(\widehat{\boldsymbol{\theta}}_{\tau})$ for the composite null hypothesis (\ref{a.12}) 
can be defined through above expressions by replacing $W_{n}^0(\boldsymbol{\theta}_{0})$ and $\chi_{p,\alpha}^2$ 
by $W_{n}(\widehat{\boldsymbol{\theta}}_{\tau})$  and $\chi_{r,\alpha}^2$  respectively,
where $\boldsymbol{\theta}_0$ is now the true null parameter in $\Theta_0$.

Let us first consider the case of simple null hypothesis and 
derive the asymptotic power under the contiguous contaminated distribution 
$\mathbf{\underline{F}}_{n,\epsilon,\mathbf{t}}^P$ in the following theorem.

\begin{theorem}
\label{THM:7asymp_power_one}
Consider the problem of testing the simple null hypothesis (\ref{a.6}) 
by the proposed Wald-type test statistics $W_{n}^0(\boldsymbol{\theta}_{0})$ 
at $\alpha$-level of significance and consider the contiguous alternative hypotheses given by (\ref{a.8}).
Then the following results hold.
	
\begin{enumerate}
\item The asymptotic distribution of $W_{n}^{0}({\boldsymbol{\theta}}_{0})$ 
under $\mathbf{\underline{F}}_{n,\epsilon,\mathbf{t}}^P$ is $\chi_{p}^2(\delta_{\epsilon})$ with  
\[
\delta_{\epsilon}= \widetilde{\boldsymbol{d}}_{\epsilon,\boldsymbol{t},{\tau}}^{T}(\boldsymbol{\theta}_{0}) 
\boldsymbol{\Sigma}_{{\tau}}^{-1}(\boldsymbol{\theta}_{0}) 
\widetilde{\boldsymbol{d}}_{\epsilon,\boldsymbol{t},{\tau}}(\boldsymbol{\theta}_{0}),
\]
where $\widetilde{\boldsymbol{d}}_{\epsilon,\boldsymbol{t},{\tau}}(\boldsymbol{\theta}_{0})
=\boldsymbol{d}+\epsilon IF(\boldsymbol{t}; \boldsymbol{T}_{\tau}, \mathbf{\underline{F}}_{{\boldsymbol{\theta}}_0})$ 
with $IF(\boldsymbol{t}; \boldsymbol{T}_{\tau}, \mathbf{\underline{F}}_{{\boldsymbol{\theta}}_0})$ 
being given by (\ref{EQ:IF_MDPDEa}).
		
\item The corresponding asymptotic power under contiguous contaminated distribution 
$\mathbf{\underline{F}}_{n,\epsilon,\mathbf{t}}^P$ is given by
\begin{align}
{\pi}_{W_{n}^{0}}(\boldsymbol{\theta}_{n},\epsilon,\boldsymbol{t})  
&=  \lim_{n \rightarrow \infty} P_{\mathbf{\underline{F}}_{n,\epsilon,\mathbf{t}}^P }(W_{n}^0(\boldsymbol{\theta}_{0})>\chi_{p,\alpha}^2)
= 1-G_{\chi_{p}^{2}(\delta_{\epsilon})}(\chi_{p,\alpha}^{2})\nonumber\\
& = \sum\limits_{v=0}^{\infty}C_{v}\left(\widetilde{\boldsymbol{d}}_{\epsilon,\boldsymbol{t},{\tau}}(\boldsymbol{\theta}_{0}),
\boldsymbol{\Sigma}_{{\tau}}^{-1}(\boldsymbol{\theta}_{0})\right)  P\left(\chi_{p+2v}^{2}>\chi_{p,\alpha}^{2}\right) ,
\label{EQ:Cont_power_one}
\end{align}
where
$
C_{v}\left(  \boldsymbol{s},\boldsymbol{A}\right)  =\frac{\left(\mathbf{s}^{T}\boldsymbol{A}\mathbf{s}\right)^{v}}{v!2^{v}}
e^{-\frac{1}{2}\mathbf{s}^{T}\boldsymbol{A}\mathbf{s}}.
$	
\end{enumerate}
\begin{proof}
Denote $\boldsymbol{\theta}_{n}^{\ast}
=\boldsymbol{T}_{\tau}(\mathbf{\underline{F}}_{n,\epsilon,\mathbf{t}}^P)$.
Then, we can express our Wald-type test statistics in (\ref{a.7}) as 
\begin{eqnarray}
&& W_{n}^0(\boldsymbol{\theta}_{0}) 
= n(\widehat{\boldsymbol{\theta}}_{\tau}-\boldsymbol{\theta}_{0})^{T}
\boldsymbol{\Sigma}_{\tau}^{-1}(\boldsymbol{\theta}_{0})(\widehat{\boldsymbol{\theta}}_{\tau}-\boldsymbol{\theta}_{0}) \nonumber\\
&  =& n(\widehat{\boldsymbol{\theta}}_{\tau}-\boldsymbol{\theta}_{n}^{\ast})^{T}
\boldsymbol{\Sigma}_{\tau}^{-1}(\boldsymbol{\theta}_{0})(\widehat{\boldsymbol{\theta}}_{\tau}-\boldsymbol{\theta}_{n}^{\ast})
+2n(\widehat{\boldsymbol{\theta}}_{\tau}-\boldsymbol{\theta}_{n}^{\ast})^{T}
\boldsymbol{\Sigma}_{\tau}^{-1}(\boldsymbol{\theta}_{0})(\boldsymbol{\theta}_{n}^{\ast}-\boldsymbol{\theta}_{0}) 
%\nonumber\\& & ~~~~~~~~~~~~~~~~~~~~~~~~~~~~~~
+ n(\boldsymbol{\theta}_{n}^{\ast}-\boldsymbol{\theta}_{0})^{T}
\boldsymbol{\Sigma}_{\tau}^{-1}(\boldsymbol{\theta}_{0})(\boldsymbol{\theta}_{n}^{\ast}-\boldsymbol{\theta}_{0}).~~~~~~
\label{W0}
\end{eqnarray}
A suitable Taylor series expansion of $\boldsymbol{\theta}_{n}^{\ast}$,  as a function of $\epsilon$ at 
$\epsilon_{n}=0$ yields \citep{Ghosh/etc:2016a} 
\begin{align*}
\boldsymbol{\theta}_{n}^{\ast}  &  =\boldsymbol{\theta}_{n}
+\tfrac{\epsilon}{\sqrt{n}}IF(\boldsymbol{t}; \boldsymbol{T}_{\tau}, \mathbf{\underline{F}}_{{\boldsymbol{\theta}}_0})
+o_p(\frac{1}{\sqrt{n}}\boldsymbol{1}_{p}) 
\end{align*}
and hence 
\begin{align*}
\sqrt{n}(\boldsymbol{\theta}_{n}^{\ast}-\boldsymbol{\theta}_{n})  &
=\epsilon IF(\boldsymbol{t}; \boldsymbol{T}_{\tau}, \mathbf{\underline{F}}_{{\boldsymbol{\theta}}_0})  + o_{p}(\boldsymbol{1}_{p}),\\
\sqrt{n}(\boldsymbol{\theta}_{n}^{\ast}-\boldsymbol{\theta}_{0}-n^{-1/2}\boldsymbol{d})  
&  =\epsilon IF(\boldsymbol{t}; \boldsymbol{T}_{\tau}, \mathbf{\underline{F}}_{{\boldsymbol{\theta}}_0})  + o_{p}(\boldsymbol{1}_{p}).
\end{align*}
Writing  $\boldsymbol{\theta}_{n}^{\ast}$ in terms of $\boldsymbol{\theta}_{0}$, we get 
\begin{align}
\sqrt{n}(\boldsymbol{\theta}_{n}^{\ast}-\boldsymbol{\theta}_{0})  &
=\boldsymbol{d}+\epsilon IF(\boldsymbol{t}; \boldsymbol{T}_{\tau}, \mathbf{\underline{F}}_{{\boldsymbol{\theta}}_0}) 
+o_{p}(\boldsymbol{1}_{p})\nonumber\\
&  =\widetilde{\boldsymbol{d}}_{\epsilon,\boldsymbol{t},\tau
}(\boldsymbol{\theta}_{0})+o_{p}(\boldsymbol{1}_{p}). \label{DT}%
\end{align}
Using these, we can rewrite (\ref{W0}) as
\begin{eqnarray}
W_{n}^0(\boldsymbol{\theta}_{0}) 
= n\left(\sqrt{n}(\widehat{\boldsymbol{\theta}}_{\tau}-\boldsymbol{\theta}_{n}^{\ast})
+\widetilde{\boldsymbol{d}}_{\epsilon,\boldsymbol{x},\tau}(\boldsymbol{\theta}_{0})\right)^{T}
\boldsymbol{\Sigma}_{\tau}^{-1}(\boldsymbol{\theta}_{0})
\left(\sqrt{n}(\widehat{\boldsymbol{\theta}}_{\tau}-\boldsymbol{\theta}_{n}^{\ast})
+\widetilde{\boldsymbol{d}}_{\epsilon,\boldsymbol{x},\tau}(\boldsymbol{\theta}_{0})\right)
+o_{p}(1).
\label{EQ:W0.1}
\end{eqnarray}
But, under $\mathbf{\underline{F}}_{n,\epsilon,\mathbf{t}}^P$, the asymptotic distribution of MDPDE 
and continuity of $\boldsymbol{\Sigma}_{\tau}(\boldsymbol{\theta})$ implies that
\begin{equation}
\sqrt{n}(\widehat{\boldsymbol{\theta}}_{\tau}-\boldsymbol{\theta}_{n}^{\ast})
\underset{n\rightarrow\infty}{\overset{\mathcal{L}}{\longrightarrow}}
\mathcal{N}(\boldsymbol{0}_{p},\boldsymbol{\Sigma}_{\tau}(\boldsymbol{\theta}_{0})). 
\label{norm}%
\end{equation}
Hence combining (\ref{EQ:W0.1}) and (\ref{norm}), we finally get
\[
W_{n}^{0}(\widehat{\boldsymbol{\theta}}_{\tau})
\underset{n\rightarrow\infty}{\overset{\mathcal{L}}{\longrightarrow}}
\chi_{p}^{2}\left(\delta_{\epsilon}\right) 
\]
with $\delta_{\epsilon}$ as defined in the statement of the theorem. 
This completes the proof of Part 1 of the theorem.

Next, Part 2 of the theorem follows directly from 
the infinite series expansion of non-central distribution functions 
in terms of that of the central chi-square variables as follows:
\begin{align}
{\pi}_{W_{n}^{0}}(\boldsymbol{\theta}_{n},\epsilon,\boldsymbol{t})  
& =   \lim_{n \rightarrow \infty} P_{\mathbf{\underline{F}}_{n,\epsilon,\mathbf{t}}^P }(W_{n}^0(\boldsymbol{\theta}_{0})>\chi_{p,\alpha}^2)
\nonumber\\
&  = P(\chi_{p}^{2}\left(  \delta_{\epsilon}\right)  >\chi_{p,\alpha}^{2})
=1-G_{\chi_{p}^{2}\left(  \delta_{\epsilon}\right)  }\left(  \chi_{p,\alpha}^{2}\right)
\nonumber\\
&  =\sum\limits_{v=0}^{\infty}C_{v}\left(\widetilde{\boldsymbol{d}}_{\epsilon,\boldsymbol{t},{\tau}}(\boldsymbol{\theta}_{0}),
\boldsymbol{\Sigma}_{{\tau}}^{-1}(\boldsymbol{\theta}_{0})\right)  P\left(\chi_{p+2v}^{2}>\chi_{p,\alpha}^{2}\right).\nonumber
\end{align}	
\end{proof}
\end{theorem}

Note that, substituting $\epsilon=0$ in Expression (\ref{EQ:Cont_power_one}) of the above theorem, 
we have an infinite series expression for the asymptotic power of the proposed Wald-type tests 
under the contiguous alternative hypotheses	(\ref{a.8}) as given by
\begin{equation}
{\pi}_{W_{n}^{0}}(\boldsymbol{\theta}_{n})
={\pi}_{W_{n}^{0}}(\boldsymbol{\theta}_{n},0,\boldsymbol{t})
= \sum\limits_{v=0}^{\infty}C_{v}\left(\boldsymbol{d},\boldsymbol{\Sigma}_{{\tau}}^{-1}(\boldsymbol{\theta}_{0})\right)  
P\left( \chi_{p+2v}^{2}>\chi_{p,\alpha}^{2}\right),\nonumber
\end{equation}
which has been previously obtained in terms of suitable distribution function in (\ref{EQ:Power_one}).

Further, substituting $\boldsymbol{d}=\boldsymbol{0}_{p}$ in Expression (\ref{EQ:Cont_power_one}), 
we get the asymptotic level of the proposed Wald-type tests under the contaminated distribution 
$\mathbf{\underline{F}}_{n,\epsilon,\mathbf{t}}^L$ as given by
\begin{align}
\alpha_{W_{n}^{0}}(\epsilon,\boldsymbol{t})  
&  ={\pi}_{W_{n}^{0}}(\boldsymbol{\theta}_{0},\epsilon,\boldsymbol{t})\big|_{\boldsymbol{d}=\boldsymbol{0}_{p}}
%\nonumber\\		&  \cong
=\sum\limits_{v=0}^{\infty}C_{v}\left(\epsilon IF(\boldsymbol{t}; \boldsymbol{T}_{\tau}, \mathbf{\underline{F}}_{{\boldsymbol{\theta}}_0}),
\boldsymbol{\Sigma}_{{\tau}}^{-1}(\boldsymbol{\theta}_{0})\right) P\left(\chi_{p+2v}^{2}>\chi_{p,\alpha}^{2}\right)  .\nonumber
\end{align}

Finally, using the expression of the asymptotic power ${\pi}_{W_{n}^{0}}(\boldsymbol{\theta}_{0},\epsilon,\boldsymbol{t})$
from (\ref{EQ:Cont_power_one}) and differentiating it suitably,  
we get the required PIF and then get the LIF by substituting $\boldsymbol{d}=\boldsymbol{0}_{p}$
as presented in the following theorem.

\begin{theorem}
\label{THM:PIF_one}
Under the assumptions of Theorem \ref{THM:7asymp_power_one},
the power and level influence functions of the proposed Wald-type tests for the simple null hypothesis (\ref{a.6}) 
are given by
\begin{equation}
PIF(\mathbf{t}; W_{n}^0, \mathbf{\underline{F}}_{{\boldsymbol{\theta}}_0} ) 
= K_{p}^{\ast}\left( \mathbf{d}^{T}\boldsymbol{\Sigma}_{{\tau}}^{-1}(\boldsymbol{\theta}_{0})\mathbf{d}\right)  
\boldsymbol{d}^{T}\boldsymbol{\Sigma}_{{\tau}}^{-1}(\boldsymbol{\theta}_{0})
 IF(\boldsymbol{t}; \boldsymbol{T}_{\tau}, \mathbf{\underline{F}}_{{\boldsymbol{\theta}}_0}),\nonumber
%\label{EQ:7PIF_simpleTest1}%
\end{equation}
with 
$
K_{p}^{\ast}(s)=e^{-\frac{s}{2}}\sum\limits_{v=0}^{\infty}\frac{s^{v-1}}{v!2^{v}}(2v-s)
P\left(\chi_{p+2v}^{2}>\chi_{p,\tau}^{2}\right)
$
and	
$$
LIF(\mathbf{t}; W_{n}^0, \mathbf{\underline{F}}_{{\boldsymbol{\theta}}_0} ) =0.
$$
\end{theorem}
\begin{proof}
Considering the expression of ${\pi}_{W_{n}^{0}}(\boldsymbol{\theta}_{0},\epsilon,\boldsymbol{t})$
from (\ref{EQ:Cont_power_one}) in Theorem \ref{THM:7asymp_power_one} and using the definition of the PIF
along with the chain rule of derivatives, we get
\begin{eqnarray}
PIF(\mathbf{t}; W_{n}^0, \mathbf{\underline{F}}_{{\boldsymbol{\theta}}_0} )  
&=& \frac{\partial}{\partial\epsilon}\left.{\pi}_{W_{n}^{0}}(\boldsymbol{\theta}_{0},\epsilon,\boldsymbol{t})\right\vert _{\epsilon=0}
\nonumber\\
& =& \sum\limits_{v=0}^{\infty}\frac{\partial}{\partial\epsilon}\left.
C_{v}\left(\widetilde{\boldsymbol{d}}_{\epsilon,\boldsymbol{t},{\tau}}(\boldsymbol{\theta}_{0}),
\boldsymbol{\Sigma}_{{\tau}}^{-1}(\boldsymbol{\theta}_{0})\right)\right\vert _{\epsilon=0}
P\left(  \chi_{p+2v}^{2}>\chi_{p,\alpha}^{2}\right) \nonumber\\
&  =&  \sum\limits_{v=0}^{\infty}\left\{  \frac{\partial}{\partial\mathbf{s}}\left.  
C_{v}\left(  \mathbf{s},\boldsymbol{\Sigma}_{\tau}^{-1}(\boldsymbol{\theta}_{0})\right) 
\right\vert _{\mathbf{s}=\widetilde{\boldsymbol{d}}_{0,\boldsymbol{t},\tau}(\boldsymbol{\theta}_{0})} \right\}^{T}
\left\{  \frac{\partial}{\partial\epsilon}\left.
\widetilde{\boldsymbol{d}}_{\epsilon,\boldsymbol{t},\tau}(\boldsymbol{\theta}_{0})\right\vert _{\epsilon=0}\right\}  
P\left(\chi_{p+2v}^{2}>\chi_{p,\alpha}^{2}\right).~~~~~
\label{pifeq}
\end{eqnarray}
Now, one can check that  
$\widetilde{\boldsymbol{d}}_{0,\boldsymbol{x},\tau}(\boldsymbol{\theta}_{0})=\boldsymbol{d}$, 
$
\frac{\partial}{\partial\epsilon}\widetilde{\boldsymbol{d}}_{\epsilon,\boldsymbol{t},\tau}(\boldsymbol{\theta}_{0})
=IF(\boldsymbol{t}; \boldsymbol{T}_{\tau}, \mathbf{\underline{F}}_{{\boldsymbol{\theta}}_0})
$
and
\[
\frac{\partial}{\partial\mathbf{t}}C_{v}\left(  \mathbf{t},\mathbf{A}\right)
=\frac{\left(  \mathbf{t}^{T}\mathbf{A}\mathbf{t}\right)  ^{v-1}}{v!2^{v}%
}\left(  2v-\mathbf{t}^{T}\mathbf{A}\mathbf{t}\right)  \mathbf{A}%
\mathbf{t}e^{-\frac{1}{2}\mathbf{t}^{T}\mathbf{A}\mathbf{t}}.
\]
Substituting these expressions in (\ref{pifeq}) and simplifying, 
we obtain the required expression of the PIF as given in the theorem.

Finally, the LIF follows from the PIF by substituting $\boldsymbol{d}=\boldsymbol{0}_{p}$.
\end{proof}

\bigskip
Note that the above PIF is clearly bounded if and only if 
the IF of the MDPDE functional $\boldsymbol{T}_{\tau}$ is bounded, i.e., whenever $\tau>0$.
This again implies the robustness of the asymptotic power of the proposed Wald-type tests with $\tau>0$ 
under contiguous contamination, over the classical MLE based Wald test (at $\tau=0$) that has an unbounded PIF. 
Further, the asymptotic level of the Wald-type tests will not also be affected by a contiguous contamination
as suggested by its zero LIF.

Next we can similarly derive the PIF and LIF for the proposed Wald-type test statistics 
$W_{n}(\widehat{\boldsymbol{\theta}}_{\tau})$ for composite null hypotheses also.
For brevity, we only present the main results corresponding to Theorem \ref{THM:7asymp_power_one}
and \ref{THM:PIF_one} for the composite hypotheses case in the following theorems;
proofs are similar and hence omitted. 
The main implications are again the same proving the claimed robustness of 
our proposal with $\tau>0$ in terms of its asymptotic level and power under contiguous contamination 
through zero LIF and bounded PIF.

\begin{theorem}
\label{THM:7asymp_power_composite}
Consider the problem of testing the composite null hypothesis (\ref{a.12}) 
by the proposed Wald-type test statistics $W_{n}(\widehat{\boldsymbol{\theta}}_{\tau})$ 
at $\alpha$-level of significance and consider the contiguous alternative hypotheses given by (\ref{a.15}).
Then the following results hold.
	
\begin{enumerate}
\item The asymptotic distribution of $W_{n}(\widehat{\boldsymbol{\theta}}_{\tau})$ 
under $\mathbf{\underline{F}}_{n,\epsilon,\mathbf{t}}^P$ is $\chi_{r}^2(\delta_{\epsilon}^\ast)$ with  
\[
\delta_{\epsilon}^\ast = 
\widetilde{\boldsymbol{d}}_{\epsilon,\boldsymbol{t},{\tau}}^{T}(\boldsymbol{\theta}_{0})
\boldsymbol{H}(\boldsymbol{\theta}_{0})
\left[\boldsymbol{H}^{T}(\boldsymbol{\theta}_{0}) \boldsymbol{\Sigma}_{{\tau}}(\boldsymbol{\theta}_{0}) \boldsymbol{H}(\boldsymbol{\theta}_{0})\right]^{-1}
\boldsymbol{H}^{T}(\boldsymbol{\theta}_{0}) 
\widetilde{\boldsymbol{d}}_{\epsilon,\boldsymbol{t},{\tau}}(\boldsymbol{\theta}_{0}),
\]
		
\item The corresponding asymptotic power under contiguous contaminated distribution 
$\mathbf{\underline{F}}_{n,\epsilon,\mathbf{t}}^P$ is given by
\begin{align}
{\pi}_{W_{n}}(\boldsymbol{\theta}_{n},\epsilon,\boldsymbol{t})  
&  \lim_{n \rightarrow \infty} P_{\mathbf{\underline{F}}_{n,\epsilon,\mathbf{t}}^P }
(W_{n}(\widehat{\boldsymbol{\theta}}_{\tau})>\chi_{r,\alpha}^2)
= 1-G_{\chi_{r}^{2}(\delta_{\epsilon}^\ast)}(\chi_{r,\alpha}^{2})\nonumber\\
& = \sum\limits_{v=0}^{\infty}C_{v}\left(
\boldsymbol{H}^{T}(\boldsymbol{\theta}_{0})\widetilde{\boldsymbol{d}}_{\epsilon,\boldsymbol{t},{\tau}}(\boldsymbol{\theta}_{0}),
\left[\boldsymbol{H}^{T}(\boldsymbol{\theta}_{0}) \boldsymbol{\Sigma}_{{\tau}}(\boldsymbol{\theta}_{0})
\boldsymbol{H}(\boldsymbol{\theta}_{0})\right]^{-1}\right)  
P\left(\chi_{r+2v}^{2}>\chi_{r,\alpha}^{2}\right).		
\label{EQ:Cont_power_composite}
\end{align}
\end{enumerate}
\end{theorem}

\begin{theorem}
\label{THM:PIF_composite}
Under the assumptions of Theorem \ref{THM:7asymp_power_composite},
the power and level influence functions of the proposed Wald-type test of the composite null hypothesis 
are given by
\begin{eqnarray}
PIF(\mathbf{t}; W_{n}, \mathbf{\underline{F}}_{{\boldsymbol{\theta}}_0} ) 
&=& K_{r}^{\ast}\left( \delta_0^\ast\right)
%\mathbf{d}^{T}\boldsymbol{H}(\boldsymbol{\theta}_{0})
%\left[\boldsymbol{H}^{T}(\boldsymbol{\theta}_{0}) \boldsymbol{\Sigma}_{{\tau}}(\boldsymbol{\theta}_{0}) \boldsymbol{H}(\boldsymbol{\theta}_{0})\right]^{-1}
%\boldsymbol{H}^{T}(\boldsymbol{\theta}_{0})\mathbf{d}\right)  
\boldsymbol{d}^{T}\boldsymbol{H}(\boldsymbol{\theta}_{0})
\left[\boldsymbol{H}^{T}(\boldsymbol{\theta}_{0}) \boldsymbol{\Sigma}_{{\tau}}(\boldsymbol{\theta}_{0}) \boldsymbol{H}(\boldsymbol{\theta}_{0})\right]^{-1}
\boldsymbol{H}^{T}(\boldsymbol{\theta}_{0}) 
IF(\boldsymbol{t}; \boldsymbol{T}_{\tau}, \mathbf{\underline{F}}_{{\boldsymbol{\theta}}_0}),\nonumber\\
%\label{EQ:7PIF_simpleTest1}%
%\end{eqnarray}
%and	
LIF(\mathbf{t}; W_{n}, \mathbf{\underline{F}}_{{\boldsymbol{\theta}}_0} ) &=& 0.
\end{eqnarray}
\end{theorem}

\section{Application: Testing for Linear Hypotheses in Generalized Linear Models (GLMs) with fixed design}
\label{sec4}

In this Section we apply the theoretical results obtained in this paper for
non-homogeneous observations to the generalized linear model (GLM). Therefore
now the density function associated to the independent random variables $%
Y_{i},$ $1\leq i\leq n,$ is given by 
\begin{equation}
f_{i,\boldsymbol{\theta}}(y_i)	= 
f\left( y_{i},\eta _{i},\phi \right) =
\exp \left \{ \frac{y_{i}\eta_{i}-b(\eta _{i})}{a(\phi )}+c\left( y_{i},\phi \right) \right \} ,\text{ }%
	1\leq i\leq n,  \label{A}
\end{equation}%
where the canonical parameter, $\eta _{i}$, is an unknown measure of localization 
depending on the given fixed design points $\boldsymbol{x}_{i}\in \mathbb{R}^{k}$, $1\leq i\leq n$ 
and $\phi $ is a known or unknown nuisance scale or dispersion parameter 
typically required to produce standard errors following Gaussian, Gamma or inverse Gaussian distributions. 
The functions $a(\phi )$, $b(\eta _{i})$ and $c\left( y_{i},\phi \right) $ are known. 
In particular, $a(\phi )$ is set to $1$ for binomial, Poisson, and negative binomial
distribution (known $\phi$) and it does not enter into the calculations for standard errors. 
The mean $\mu _{i}$ of $Y_{i}$ is given by  
$\mu _{i}=\mu \left( \eta _{i}\right) =\mathrm{E}\left[ Y_{i}\right]=b^{\prime }(\eta _{i})$
%\begin{equation*}
%	\mu _{i}=\mu \left( \theta _{i}\right) =\mathrm{E}\left[ Y_{i}\right]
%	=b^{\prime }(\theta _{i})
%\end{equation*}%
and the variance by  
$\sigma _{i}^{2}=\sigma ^{2}(\eta _{i},\phi)=\mathrm{Var}\left[ Y_{i}\right]=a(\phi )b^{\prime\prime}(\eta _{i}).$
%\begin{equation*}
%	\sigma _{i}^{2}=\sigma ^{2}(\theta _{i})=\mathrm{Var}\left[ Y_{i}\right]
%	=a(\phi )b^{\prime \prime }(\theta _{i}).
%\end{equation*}
The mean response is assumed, according to GLMs, to be modeled linearly with respect to $\boldsymbol{x}_i$ 
through a known link function, $g$, i.e., 	$g(\mu_i)=\boldsymbol{x}_i^{T}\boldsymbol{\beta }$,
%\begin{equation*}
%	g(\theta )=\boldsymbol{x}^{T}\boldsymbol{\beta },
%\end{equation*}%
where $g$ is a monotone and differentiable function and $\boldsymbol{\beta\in }\mathbb{R}^{k}$ is an unknown parameter. 
In this setting, since $\eta_{i}=\eta _{i}\left( \boldsymbol{x}_{i}^{T}\boldsymbol{\beta }\right) $,
we shall also denote (\ref{A}) by $f\left( y_{i},\boldsymbol{x}_{i}^{T}\boldsymbol{\beta },\phi \right) $ 
and the common parameters of the GLM by 
%\begin{equation*}
$\boldsymbol{\theta }=(\boldsymbol{\beta }^{T},\phi )^{T}$, $p=k+1$.
%\end{equation*}%
At the same time we denote by 
$\widehat{\boldsymbol{\theta }}_{\tau}$=$\left(\widehat{\boldsymbol{\beta }}_{\tau}^{T},\widehat{\phi }_{\tau}\right)^{T}$ 
the minimum density power divergence estimator of $\boldsymbol{\theta }$ with tuning parameter $\tau$. 
The estimating equations, based on (\ref{a.2}), to get  $\widehat{\boldsymbol{\theta }}_{\tau }$ 
in this present case are given by%
\begin{equation}
\sum \limits_{i=1}^{n}\left[ \gamma _{1,\tau }(\boldsymbol{x}_{i})-K_{1}(y_{i},\boldsymbol{x}_{i}^{T}\boldsymbol{\beta },\phi )
f^{\tau}(y_{i},\boldsymbol{x}_{i}^{T}\boldsymbol{\beta },\phi )\right] \boldsymbol{x}_{i}=\boldsymbol{0},  \label{B}
\end{equation}%
and 
\begin{equation}
	\sum \limits_{i=1}^{n}\left[ \gamma _{2,\tau }(\boldsymbol{x}%
	_{i})-K_{2}(y_{i},\boldsymbol{x}_{i}^{T}\boldsymbol{\beta },\phi )f^{\tau
	}(y_{i},\boldsymbol{x}_{i}^{T}\boldsymbol{\beta },\phi )\right] =0.
	\label{C}
\end{equation}%
where  
\begin{align*}
K_{1}(y_{i},\boldsymbol{x}_{i}^{T}\boldsymbol{\beta },\phi )& 
=\frac{y_{i}-\mu(\eta_{i})}{\sigma^{2}(\eta_{i})g^{\prime }\left(\mu(\theta_{i})\right)}, 
%	\\
~~	K_{2}(y_{i},\boldsymbol{x}_{i}^{T}\boldsymbol{\beta },\phi ) 
=-\frac{y_{i}\eta_{i}-b\left( \eta_{i}\right)}{a^{2}(\phi)}a^{\prime }(\phi)+\frac{\partial c\left(y_{i},\phi\right)}{\partial \phi}.
\end{align*}%
and 
\begin{equation}
	\gamma _{j,\tau }(\boldsymbol{x}_{i})=\int K_{j}(y,\boldsymbol{x}_{i}^{T}%
	\boldsymbol{\beta },\phi )f^{1+\tau }(y,\boldsymbol{x}_{i}^{T}\boldsymbol{%
		\beta },\phi )dy,\quad \text{for }j=1,2.  \label{D}
\end{equation}

If we want to ignore the parameter $\phi $ and to estimate $\boldsymbol{%
	\beta }$ taking $\phi $ fixed (or, substituted suitably), it is enough to
consider only the set of estimating equations in (\ref{B}). Further, $\tau =0$, we have 
$\gamma _{1,0 }(\boldsymbol{x}_{i})=0$
%\begin{equation*}
%	\int K_{1}(y,\boldsymbol{x}_{i}^{T}\boldsymbol{\beta },\phi )f(y,%
%	\boldsymbol{x}_{i}^{T}\boldsymbol{\beta },\phi )dy=\boldsymbol{0}
%\end{equation*}%
and the estimating equations for $\boldsymbol{\beta }$ are 
\begin{equation*}
\sum \limits_{i=1}^{n}\frac{y_{i}-\mu(\eta_{i})}{\sigma^{2}(\eta_{i})g^{\prime}\left(\mu(\eta_{i})\right)}\boldsymbol{x}_{i}
=\boldsymbol{0}.
\end{equation*}

The asymptotic distribution of $\widehat{\boldsymbol{\theta }}_{\tau}$ is then given by (\ref{a.25}), 
%\begin{equation}
%	\sqrt{n}(\widehat{\boldsymbol{\eta }}_{\tau }-\boldsymbol{\eta }_{0})%
%	\underset{n\rightarrow \infty }{\longrightarrow }\mathcal{N}(\boldsymbol{0},%
%	\boldsymbol{\Sigma }_{\tau }(\boldsymbol{\eta }_{0})),  \label{F}
%\end{equation}%
%where 
%\begin{equation*}
%	\boldsymbol{\Sigma }_{\tau }(\boldsymbol{\eta }_{0})=\lim_{n\rightarrow
%		\infty }\boldsymbol{J}_{\tau }^{-1}(\boldsymbol{\eta }_{0})\boldsymbol{K}%
%	_{\tau }(\boldsymbol{\eta }_{0})\boldsymbol{J}_{\tau }^{-1}(\boldsymbol{\eta 
%	}_{0}),
%\end{equation*}%
where we now have
\begin{equation*}
\boldsymbol{\Omega}_{n,\tau }(\boldsymbol{\theta })=%
\begin{pmatrix}
	\sum \limits_{i=1}^{n}\left( \gamma _{11,2\tau }(\boldsymbol{x}_{i})-\gamma
		_{1,\tau }^{2}(\boldsymbol{x}_{i})\right) \boldsymbol{x}_{i}\boldsymbol{x}%
		_{i}^{T} & \sum \limits_{i=1}^{n}\left( \gamma _{12,2\tau }(\boldsymbol{x}%
		_{i})-\gamma _{1,\tau }(\boldsymbol{x}_{i})\gamma _{1,\tau }(\boldsymbol{x}%
		_{i})\right) \boldsymbol{x}_{i} \\ 
		\sum \limits_{i=1}^{n}\left( \gamma _{12,2\tau }(\boldsymbol{x}_{i})-\gamma
		_{1,\tau }(\boldsymbol{x}_{i})\gamma _{1,\tau }(\boldsymbol{x})\right) 
		\boldsymbol{x}_{i}^{T} & \sum \limits_{i=1}^{n}\left( \gamma _{22,2\tau }(%
		\boldsymbol{x}_{i})-\gamma _{2,\tau }^{2}(\boldsymbol{x}_{i})\right) 
	\end{pmatrix}%
	,
\end{equation*}%
and  
\begin{equation*}
	\boldsymbol{\Psi}_{n,\tau }(\boldsymbol{\theta })=%
	\begin{pmatrix}
		\sum \limits_{i=1}^{n}\gamma _{11,\tau }(\boldsymbol{x}_{i})\boldsymbol{x}%
		_{i}\boldsymbol{x}_{i}^{T} & \sum \limits_{i=1}^{n}\gamma _{12,\tau }(%
		\boldsymbol{x}_{i})\boldsymbol{x}_{i} \\ 
		\sum \limits_{i=1}^{n}\gamma _{12,\tau }(\boldsymbol{x}_{i})\boldsymbol{x}%
		_{i} & \sum \limits_{i=1}^{n}\gamma _{22,\tau }(\boldsymbol{x}_{i})%
	\end{pmatrix}%
	,
\end{equation*}%
with $\gamma _{j,\tau }(\boldsymbol{x})$, $j=1,2$, being given by (\ref{D})
and 
\begin{equation*}
	\gamma _{jh,\tau }(\boldsymbol{x})=\int K_{j}\left( y,\boldsymbol{x}^{T}%
	\boldsymbol{\beta },\phi \right) K_{h}\left( y,\boldsymbol{x}^{T}\boldsymbol{%
		\beta },\phi \right) f^{1+\tau }\left( y,\boldsymbol{x}^{T}\boldsymbol{\beta 
	},\phi \right) dy\text{, for }j,h=1,2.
\end{equation*}%
Notice that for the case where $\phi $ is known we get 
\begin{equation}
\boldsymbol{\Omega}_{n,\tau }(\boldsymbol{\theta })=\sum \limits_{i=1}^{n}\left(
	\gamma _{11,\tau }(\boldsymbol{x}_{i})-\gamma _{1,\tau }^{2}(\boldsymbol{x}%
	_{i})\right) \boldsymbol{x}_{i}\boldsymbol{x}_{i}^{T}, 
%\end{equation*}
~\mbox{and}~ 
%\begin{equation*}
\boldsymbol{\Psi}_{n,\tau }(\boldsymbol{\theta })=\sum \limits_{i=1}^{n}\gamma_{11,\tau }(\boldsymbol{x}_{i})
\boldsymbol{x}_{i}\boldsymbol{x}_{i}^{T}.
\label{EQ:mat_betaO}
\end{equation}
See \cite{Ghosh/Basu:2016} for more details on the properties of the MDPDE 
$\widehat{\boldsymbol{\theta }}_{\tau}$ under this fixed-design GLM.

Here, we consider the most important hypothesis testing problem in the context of GLM, 
namely testing the linear hypothesis on regression coefficient $\boldsymbol{\beta}$ as given by 
\begin{equation}
H_{0}:\boldsymbol{L}\boldsymbol{\beta } =\boldsymbol{l}_0\text{
	versus }H_{1}:\boldsymbol{L}\boldsymbol{\beta } \neq \boldsymbol{l}_0,  
\label{G}
\end{equation}%
with $\boldsymbol{L}$ being an $r\times k$ known matrix of rank $r$ and $\boldsymbol{l}_0$ being an known $r$-vector with $r\leq k$.
Note that, this particular hypothesis (\ref{G}) belongs to the general class of hypothesis in (\ref{a.12}) with 
$\boldsymbol{h}\left( \boldsymbol{\eta }\right) =\boldsymbol{L}\boldsymbol{\beta } -\boldsymbol{l}_0$
and $\boldsymbol{H}\left( \boldsymbol{\eta }\right) =\boldsymbol{L}^T$.
Now for testing  (\ref{G}), we can consider the family of Wald-type test statistics presented in Section \ref{sec2.2}, 
given by   
\begin{equation}
W_{n}(\widehat{\boldsymbol{\theta }}_{\tau })
%=n\boldsymbol{h}\left( \widehat{\boldsymbol{\theta }}_{\tau }\right)^{T}
%\left[ \boldsymbol{H}(\widehat{\boldsymbol{\theta }}_{\tau })^{T}\boldsymbol{\Sigma }_{\tau}(\widehat{\boldsymbol{\theta }}_{\tau})
%\boldsymbol{H}(\widehat{\boldsymbol{\theta }}_{\tau })\right] ^{-1}
%\boldsymbol{h}\left( \widehat{\boldsymbol{\eta }}_{\tau }\right) 
=n\left( \boldsymbol{L}\widehat{\boldsymbol{\beta}}_{\tau} -\boldsymbol{l}_0\right)^{T}
\left[ \boldsymbol{L}\boldsymbol{\Sigma }_{\tau}(\widehat{\boldsymbol{\theta }}_{\tau})\boldsymbol{L}^T\right]^{-1}
\left(\boldsymbol{L}\widehat{\boldsymbol{\beta}}_{\tau} -\boldsymbol{l}_0\right).  \label{H}
\end{equation}

Based on our Theorem \ref{THM:null_composite}, 
the null hypothesis given in (\ref{G}) will be rejected if we have that 
\begin{equation}
	W_{n}(\widehat{\boldsymbol{\theta }}_{\tau })>\chi _{r,\alpha }^{2}.  \label{I}
\end{equation}
Further, following discussions in Section \ref{sec2.2}, this proposed Wald-type test is consistent at any fixed alternatives
and one can obtain an approximation to its power function at any fixed alternatives.

Next, suppose the true null parameter value is $\boldsymbol{\theta}_0 = \left(\boldsymbol{\beta}_0^T, \phi_0\right)^T$
and consider the sequence of contiguous alternatives 
$H_{1,n} : \boldsymbol{\beta}_n = \boldsymbol{\beta}_0 + n^{-1/2}\boldsymbol{d}$ 
with $\boldsymbol{d}\in\mathbb{R}^k - \{\boldsymbol{0}\}$. 
This is also equivalent to the alternative contiguous hypothesis   
$H_{1,n} : \boldsymbol{L}\boldsymbol{\beta}_n = \boldsymbol{l}_0 + n^{-1/2}\boldsymbol{d}^\ast$ 
with $\boldsymbol{d}^\ast=\boldsymbol{L}\boldsymbol{d}\in\mathbb{R}^r - \{\boldsymbol{0}\}$. 
Under these contiguous alternatives, the proposed Wald-type test statistics 
have the asymptotic distribution as non-central chi-square with degrees of freedom $r$ and 
non-centrality parameter 
\begin{equation}
\delta=\boldsymbol{d}^T\boldsymbol{L}^T
\left[ \boldsymbol{L}\boldsymbol{\Sigma }_{\tau}(\boldsymbol{\theta}_0)\boldsymbol{L}^T\right]^{-1}
\boldsymbol{L}\boldsymbol{d} =\boldsymbol{d}^{\ast T}
\left[ \boldsymbol{L}\boldsymbol{\Sigma }_{\tau}(\boldsymbol{\theta}_0)\boldsymbol{L}^T\right]^{-1}\boldsymbol{d}^\ast.
\label{EQ:delta_glm}
\end{equation}
Then, the asymptotic power at these contiguous alternatives can easily be obtained 
through the upper cumulative distribution functions of the above non-central chi-square distributions.
We will examine their behavior for some special cases of GLM in the next section.

Next, considering the robustness of the proposed Wald-type tests under GLM, 
the first order influence function of the test statistics and the level influence functions are always identically zero 
under contamination in any fixed direction or in all directions following the general theory developed in Section \ref{sec3}.
For both types of contaminations, the non-zero second order influence function of the 
proposed Wald-type test statistics (\ref{H}) under fixed-design GLM is given by 
\begin{eqnarray}
IF_{i_0}^{(2)}(t_{i_0}; W_{\tau}, \mathbf{\underline{F}}_{{\boldsymbol{\theta}}_0}) 
&=& 2 IF_{i_0}(t_{i_0}; \boldsymbol{T}_{\tau}, \mathbf{\underline{F}}_{{\boldsymbol{\theta}}_0})^T 
\boldsymbol{L}^T\left[ \boldsymbol{L}\boldsymbol{\Sigma }_{\tau}(\boldsymbol{\theta}_0)\boldsymbol{L}^T\right]^{-1}
\boldsymbol{L}IF_{i_0}(t_{i_0}; \boldsymbol{T}_{\tau}, \mathbf{\underline{F}}_{{\boldsymbol{\theta}}_0}) \nonumber\\
%\label{EQ:IF2W0}\\
IF^{(2)}(\boldsymbol{t}; W_{\tau}, \mathbf{\underline{F}}_{{\boldsymbol{\theta}}_0}) 
&=& 2 IF(\boldsymbol{t}; \boldsymbol{T}_{\tau}, \mathbf{\underline{F}}_{{\boldsymbol{\theta}}_0})^T 
\boldsymbol{L}^T\left[ \boldsymbol{L}\boldsymbol{\Sigma }_{\tau}(\boldsymbol{\theta}_0)\boldsymbol{L}^T\right]^{-1}
\boldsymbol{L}IF(\boldsymbol{t}; \boldsymbol{T}_{\tau}, \mathbf{\underline{F}}_{{\boldsymbol{\theta}}_0}),\nonumber
%\label{EQ:IF2Wa}
\end{eqnarray}
where $IF_{i_0}(t_{i_0}; \boldsymbol{T}_{\tau}, \mathbf{\underline{F}}_{{\boldsymbol{\theta}}_0})$
and $IF(\boldsymbol{t}; \boldsymbol{T}_{\tau}, \mathbf{\underline{F}}_{{\boldsymbol{\theta}}_0})$
are corresponding influence functions of the MDPDE $\widehat{\boldsymbol{\theta}}_{\tau}$ under the fixed-design GLM
for contamination in the $i_0$-th direction and all directions respectively. 
These influence functions of the MDPDE under fixed-design GLM have been studied by \cite{Ghosh/Basu:2016};
using the explicit form of these IFs, the second order IFs of our test statistics become
\begin{eqnarray}
IF_{i_0}^{(2)}(t_{i_0}; W_{\tau}, \mathbf{\underline{F}}_{({\boldsymbol{\beta}}_0,\phi_0)}) 
&=& 2 \left[\frac{1}{n}\boldsymbol{S}_{\tau, i_0}(t_{i_0};{\boldsymbol{\beta}}_0,\phi_0)\right]^T 
\boldsymbol{L}_{0,\tau}^\ast
%\boldsymbol\Psi_{n,\tau}^{-1}({\boldsymbol{\beta}}_0,\phi_0)
%\boldsymbol{L}^T\left[ \boldsymbol{L}\boldsymbol{\Sigma }_{\tau}({\boldsymbol{\beta}}_0,\phi_0)\boldsymbol{L}^T\right]^{-1}
%\boldsymbol{L}\boldsymbol\Psi_{n,\tau}^{-1}({\boldsymbol{\beta}}_0,\phi_0)
\left[\frac{1}{n}\boldsymbol{S}_{\tau, i_0}(t_{i_0};{\boldsymbol{\beta}}_0,\phi)\right], \nonumber\\
%\label{EQ:IF2W0}\\
IF^{(2)}(\boldsymbol{t}; W_{\tau}, \mathbf{\underline{F}}_{({\boldsymbol{\beta}}_0,\phi_0)}) 
&=& 2 \left[\frac{1}{n}\sum_{i=1}^{n}\boldsymbol{S}_{\tau, i}(t_{i};{\boldsymbol{\beta}}_0,\phi_0)\right]^T 
\boldsymbol{L}_0^\ast
%\boldsymbol\Psi_{n,\tau}^{-1}({\boldsymbol{\beta}}_0,\phi_0)\boldsymbol{L}^T
%\left[ \boldsymbol{L}\boldsymbol{\Sigma }_{\tau}({\boldsymbol{\beta}}_0,\phi_0)\boldsymbol{L}^T\right]^{-1}
%\boldsymbol{L}\boldsymbol\Psi_{n,\tau}^{-1}({\boldsymbol{\beta}}_0,\phi_0)
\left[\frac{1}{n}\sum_{i=1}^{n}\boldsymbol{S}_{\tau, i}(t_{i};{\boldsymbol{\beta}}_0,\phi_0)\right],\nonumber
%\label{EQ:IF2Wa}
\end{eqnarray}
where
$\boldsymbol{L}_{0,\tau}^\ast = \boldsymbol\Psi_{n,\tau}^{-1}({\boldsymbol{\beta}}_0,\phi_0)\boldsymbol{L}^T
\left[ \boldsymbol{L}\boldsymbol{\Sigma }_{\tau}({\boldsymbol{\beta}}_0,\phi_0)\boldsymbol{L}^T\right]^{-1}
\boldsymbol{L}\boldsymbol\Psi_{n,\tau}^{-1}({\boldsymbol{\beta}}_0,\phi_0)$ and 
\begin{eqnarray}
\boldsymbol{S}_{\tau, i}(t_{i};{\boldsymbol{\beta}},\phi) = 
\begin{bmatrix}
\begin{array}{c}
\left( K_{1}(t_{i},\boldsymbol{x}_{i}^{T}\boldsymbol{\beta },\phi )
f^{\tau}(t_{i},\boldsymbol{x}_{i}^{T}\boldsymbol{\beta },\phi ) - \gamma _{1,\tau }(\boldsymbol{x}_{i})\right) \boldsymbol{x}_{i}\\\\
\left( K_{2}(t_{i},\boldsymbol{x}_{i}^{T}\boldsymbol{\beta },\phi )
f^{\tau}(t_{i},\boldsymbol{x}_{i}^{T}\boldsymbol{\beta },\phi ) - \gamma _{2,\tau }(\boldsymbol{x}_{i})\right)
\end{array}
\end{bmatrix}.
\end{eqnarray}
Clearly these influence functions will be bounded whenever the function 
$\boldsymbol{S}_{\tau, i}(t_{i};{\boldsymbol{\beta}}_0,\phi_0)$ is bounded in $t_i$. 
However, due to the particular exponential form of the density in GLM, we have 
$K_{j}(t_{i},\boldsymbol{x}_{i}^{T}\boldsymbol{\beta },\phi )$ is a polynomial function of $t_i$ and  
the integral  $\gamma _{j,\tau }(\boldsymbol{x}_{i})$ is bounded for any given finite $\boldsymbol{x}_i$ for each $j=1,2$. 
Hence, for any $\tau>0$, the function $\boldsymbol{S}_{\tau, i}(t_{i};{\boldsymbol{\beta}}_0,\phi_0)$ will bounded in $t_i$
and it will be unbounded at $\tau=0$. This implies that the proposed  Wald-type tests with $\tau>0$ under fixed-design  GLM
will be robust compared to the non-robust classical Wald-test at $\tau=0$.

We can similarly also check the power robustness of our proposal at $\tau>0$ under fixed-design GLM
by deriving the form of PIF from Theorem \ref{THM:PIF_composite}, 
whose boundedness again depends directly on the boundedness of the function 
$\boldsymbol{S}_{\tau, i}(t_{i};{\boldsymbol{\beta}}_0,\phi_0)$ with respect to the contamination points $t_i$s.
In particular, the form of the PIF under contiguous alternatives $H_{1,n}$ 
for the present case of fixed-design GLM simplifies to
\begin{eqnarray}
PIF(\mathbf{t}; W_{n}, \mathbf{\underline{F}}_{({\boldsymbol{\beta}}_0,\phi_0)} ) 
&=& K_{r}^{\ast}\left( \delta\right)
\boldsymbol{d}^{T}\boldsymbol{L}^T
\left[ \boldsymbol{L}\boldsymbol{\Sigma }_{\tau}({\boldsymbol{\beta}}_0,\phi_0)\boldsymbol{L}^T\right]^{-1}
\boldsymbol{L}\left[\frac{1}{n}\sum_{i=1}^{n}\boldsymbol{S}_{\tau, i}(t_{i};{\boldsymbol{\beta}}_0,\phi_0)\right],\nonumber
\end{eqnarray}
where $\delta$ is as given by Equation (\ref{EQ:delta_glm}).
We will further study the behavior of these influence functions for some 
particular examples of GLM in the next section.

It is worthwhile to note that the GLM considered in this paper 
as an special case of general non-homogeneous set-up is different from 
the usual GLM with stochastic covariates (random design);
here we are assuming that the values of covariates (design-points)
$\boldsymbol{x}_i$ are fixed and known previously. 
The problem of robust hypothesis testing under GLM with random design 
has been considered in \cite{Basu/etc:2016b,Basu/etc:2016c}.

\section{Examples and Illustrations\label{sec5}}

\subsection{Testing Significance of a Normal Linear Regression Model}\label{sec5.1}

As our first illustrative example, we will consider the most common and simplest 
case of GLM, namely the normal regression model where the model density $f(y_i, \boldsymbol{x}_i^T\boldsymbol{\beta}, \phi)$
is normal with mean $\boldsymbol{x}_i^T\boldsymbol{\beta}$ and common variance $\phi>0$. 
This model has a simpler representation given by
$$
y_i = \boldsymbol{x}_i^T\boldsymbol{\beta} + \varepsilon_i, ~~i=1, \ldots, n,
$$
where $\varepsilon_i$s are independent normally distributed errors with mean $0$ and variance $\phi$. 
When the design points $\boldsymbol{x}_i$s are pre-fixed, we can apply the results derived above 
to construct and study robust Wald-type tests for general linear hypothesis under this simpler model.
In particular, for illustration, let us consider the problem of testing for the significance of 
this linear model characterized by the hypothesis
\begin{eqnarray}
H_{0}:\boldsymbol{\beta } =\boldsymbol{\beta}_0
\text{ versus }
H_{1}:\boldsymbol{\beta } \neq \boldsymbol{\beta}_0,  
\label{EQ:lrm_hyp}
\end{eqnarray}
where $\boldsymbol{\beta}_0$ is a known $k$-vector of hypothesized regression coefficients (usually a zero vector)
and we assume $\phi$ to be unknown under both hypotheses. 
The classical F-test for this problem is a version of the classical Wald test based on the MLE of the parameters 
$\boldsymbol{\theta} = (\boldsymbol{\beta}^T, \phi)^T$ and hence known to be highly non-robust.
We will now study the performances of the proposed Wald-type tests for this hypothesis.

Note that the hypothesis in (\ref{EQ:lrm_hyp}) under the normal linear regression model 
belongs to the class of general linear hypothesis with $\boldsymbol{L}=\begin{pmatrix}
\begin{array}{cc}
\boldsymbol{I}_k & \boldsymbol{0}\\
\boldsymbol{0} & 0
\end{array}
\end{pmatrix}$ with 
$\boldsymbol{I}_k$ being the identity matrix of order $k$ 
and $\boldsymbol{l}_0=(\boldsymbol{\beta}_0^T 0)^T$.
So, using the results of the previous section, the robust Wald-type test statistics for 
this testing problem  simplifies to 
\begin{equation}
W_{n}(\widehat{\boldsymbol{\beta }}_{\tau},\widehat{\phi}_{\tau})
=n\left( \widehat{\boldsymbol{\beta}}_{\tau} -\boldsymbol{\beta}_0\right)^{T}
\left[ \boldsymbol{L}\boldsymbol{\Sigma }_{\tau}(\widehat{\boldsymbol{\beta }}_{\tau},\widehat{\phi}_{\tau})\boldsymbol{L}^T\right]^{-1}
\left(\widehat{\boldsymbol{\beta}}_{\tau} -\boldsymbol{\beta}_0\right),  
\label{EQ:lrm_TS}
\end{equation}
where $\widehat{\boldsymbol{\beta }}_{\tau}$ and $\widehat{\phi}_{\tau}$ are the MDPDE of $\boldsymbol{\beta}$ and $\phi$
respectively with tuning parameter $\tau$ and have asymptotic joint covariance matrix $\boldsymbol{\Sigma}_{\tau}(\boldsymbol{\beta}_0,\phi_0)$
at the true null parameter values $(\boldsymbol{\beta}_0,\phi_0)$. 
\cite{Ghosh/Basu:2013} studied the properties of these MDPDEs under fixed design normal linear model in detail.
In particular it follows that, under assumptions (R1)--(R2) of their paper (also listed in Appendix \ref{APP:cond}),
asymptotically 
$$\sqrt{n}\left((\widehat{\boldsymbol{\beta }}_{\tau}^T,\widehat{\phi}_{\tau})^T - (\boldsymbol{\beta}_0^T,\phi_0)^T\right)$$
follows a $k+1$-variate normal distribution with mean vector $\boldsymbol{0}$
and covariance matrix given by 
$$
\boldsymbol{\Sigma}_{\tau}(\boldsymbol{\beta}_0,\phi_0)= 
\begin{bmatrix}
\begin{array}{cc}
\upsilon_{\tau}^{\beta} \boldsymbol{C}_x^{-1} & \boldsymbol{0}\\
\boldsymbol{0}^T & \upsilon_{\tau}^{\phi}
\end{array}
\end{bmatrix},
$$
where $\boldsymbol{C}_x = \displaystyle\lim\limits_{n\rightarrow\infty}\frac{1}{n}(\boldsymbol{X}^T\boldsymbol{X})$
with $\boldsymbol{X}^T=\left[\boldsymbol{x}_1 \cdots \boldsymbol{x}_n\right]$ being the design matrix and
$$\upsilon_{\tau}^{\beta}= \phi \left(1+\frac{\tau^2}{1+2\tau}\right)^{3/2},~~
%$ and $
\upsilon_{\tau}^{\phi} = \frac{4\phi^2}{(2+\tau^2)^2}\left[2(1+2\tau^2)\left(1+\frac{\tau^2}{1+2\tau}\right)^{5/2}-\tau^2(1+\tau)^2\right].
$$ 
Using these expressions, our proposed Wald-type test statistics (\ref{EQ:lrm_TS}) for testing (\ref{EQ:lrm_hyp})
further simplifies to 
\begin{equation}
W_{n}(\widehat{\boldsymbol{\beta }}_{\tau},\widehat{\phi}_{\tau})
=\frac{n}{\widehat{\phi}_{\tau}}\left(1+\frac{\tau^2}{1+2\tau}\right)^{-3/2}
\left( \widehat{\boldsymbol{\beta}}_{\tau} -\boldsymbol{\beta}_0\right)^{T}\boldsymbol{C}_x
\left(\widehat{\boldsymbol{\beta}}_{\tau} -\boldsymbol{\beta}_0\right),  
\label{EQ:lrm_TS1}
\end{equation}
which coincides with the classical Wald test at $\tau=0$.
Following the theory of Section \ref{sec4}, these Wald-type test statistics have asymptotic null distribution as
$\chi^2_k$ and consistent against any fixed alternatives. 
To study its power against contiguous alternatives  
$H_{1,n} : \boldsymbol{\beta}_n = \boldsymbol{\beta}_0 + n^{-1/2}\boldsymbol{d}$ 
with $\boldsymbol{d}\in\mathbb{R}^k - \{\boldsymbol{0}\}$, note that
the asymptotic distribution of the proposed Wald-type test statistics under $H_{1,n}$
is non-central $\chi^2$ with degrees of freedom $r$ and non-centrality parameter 
\begin{equation}
\delta=\frac{1}{\phi_0}\left(1+\frac{\tau^2}{1+2\tau}\right)^{-3/2}\left[\boldsymbol{d}^T\boldsymbol{C}_x\boldsymbol{d}\right]. 
\label{EQ:delta_lrm}
\end{equation}
Clearly the asymptotic contiguous power of our proposed test statistics depends on the given 
fixed values of design points through the quantity $d_x=\left[\boldsymbol{d}^T\boldsymbol{C}_x\boldsymbol{d}\right]$
along with the tuning parameter $\tau$. Table \ref{TAB:lrm_Eff} presents the empirical values of these contiguous powers
over $\tau$ for different values of $d_x$, with $\phi_0=1$ and $5\%$ level of significance. 
Note that, as the number ($k$) of regressors to be tested increases, 
we need larger values of $d_x$ to achieve any fixed values of the contiguous power; 
for a given fixed design this corresponds to larger values of $||\boldsymbol{d}||$.
Further, for any fixed $\tau$ the values of contiguous power increases as $d_x$ increases 
but for any fixed $d_x>0$ it decreases as $\tau$ increases as expected;
the choice $d_x=0$ leads to the level of the tests for all $\tau \geq 0$.
However, interestingly, the loss in power compared to the classical Wald test at $\tau=0$ 
is not quite significant at small values of $\tau>0$. 
And, against this relatively small price, we will gain substantially robustness 
against contamination in data as illustrated below with the specific forms of the influence functions.

%\begin{table}[h]
%	\caption{Contiguous power of the proposed Wald-type test for testing (\ref{EQ:lrm_hyp}) under the normal regression model}
%	\centering
%%	\resizebox{\textwidth}{!}{ %
%		\begin{tabular}{l|l| ccccccc} \hline
%		$k$ &	$d_x$ & \multicolumn{7}{|c}{$\alpha$}\\
%&				&	0	&	0.05	&	0.1	&	0.3	&	0.5	&	0.7	&	1	\\\hline
%1&			0	&	0.050	&	0.050	&	0.050	&	0.050	&	0.050	&	0.050	&	0.050	\\
%&			2	&	0.293	&	0.292	&	0.290	&	0.274	&	0.254	&	0.234	&	0.207	\\
%&			5	&	0.609	&	0.607	&	0.603	&	0.574	&	0.535	&	0.494	&	0.437	\\
%&			10	&	0.885	&	0.884	&	0.882	&	0.859	&	0.825	&	0.786	&	0.722	\\
%&			15	&	0.972	&	0.972	&	0.971	&	0.961	&	0.944	&	0.921	&	0.877	\\
%&			20	&	0.994	&	0.994	&	0.994	&	0.990	&	0.984	&	0.973	&	0.950	\\
%&			25	&	0.999	&	0.999	&	0.999	&	0.998	&	0.996	&	0.992	&	0.981	\\
%&			50	&	1.000	&	1.000	&	1.000	&	1.000	&	1.000	&	1.000	&	1.000	\\			\hline
%20 &	0	&	0.050	&	0.050	&	0.050	&	0.050	&	0.050	&	0.050	&	0.050	\\
%&	2	&	0.096	&	0.096	&	0.096	&	0.092	&	0.088	&	0.083	&	0.078	\\
%&	5	&	0.193	&	0.193	&	0.191	&	0.179	&	0.164	&	0.150	&	0.133	\\
%&	10	&	0.402	&	0.400	&	0.397	&	0.367	&	0.331	&	0.296	&	0.252	\\
%&	15	&	0.611	&	0.609	&	0.604	&	0.565	&	0.513	&	0.461	&	0.391	\\
%&	20	&	0.775	&	0.773	&	0.768	&	0.730	&	0.675	&	0.616	&	0.531	\\
%&	25	&	0.883	&	0.881	&	0.878	&	0.847	&	0.800	&	0.745	&	0.657	\\
%&	30	&	0.944	&	0.943	&	0.941	&	0.920	&	0.885	&	0.840	&	0.761	\\\hline
%	\end{tabular}
%%}
%\label{TAB:lrm_Eff}
%\end{table}

\begin{table}[h]
	\caption{Contiguous power of the proposed Wald-type test for testing (\ref{EQ:lrm_hyp}) under the normal regression model}
	\centering
	%	\resizebox{\textwidth}{!}{ %
	\begin{tabular}{l| cccccc| cccccc} \hline
	& \multicolumn{6}{|c}{$k=1$}& \multicolumn{6}{|c}{$k=20$}\\\hline
			 & \multicolumn{6}{|c}{$\tau$}& \multicolumn{6}{|c}{$\tau$}\\
$d_x$	&	0	&	0.1	&	0.3	&	0.5	&	0.7	&	1 &	0		&	0.1	&	0.3	&	0.5	&	0.7	&	1	\\\hline
0	&	0.050	&	0.050	&	0.050	&	0.050	&	0.050	&	0.050	&	0.050	&	0.050	&	0.050	&	0.050	&	0.050	&	0.050	\\
2	&	0.293	&	0.290	&	0.274	&	0.254	&	0.234	&	0.207	&	0.096	&	0.096	&	0.092	&	0.088	&	0.083	&	0.078	\\
5	&	0.609	&	0.603	&	0.574	&	0.535	&	0.494	&	0.437	&	0.193	&	0.191	&	0.179	&	0.164	&	0.150	&	0.133	\\
10	&	0.885	&	0.882	&	0.859	&	0.825	&	0.786	&	0.722	&	0.402	&	0.397	&	0.367	&	0.331	&	0.296	&	0.252	\\
15	&	0.972	&	0.971	&	0.961	&	0.944	&	0.921	&	0.877	&	0.611	&	0.604	&	0.565	&	0.513	&	0.461	&	0.391	\\
20	&	0.994	&	0.994	&	0.990	&	0.984	&	0.973	&	0.950	&	0.775	&	0.768	&	0.730	&	0.675	&	0.616	&	0.531	\\
25	&	0.999	&	0.999	&	0.998	&	0.996	&	0.992	&	0.981	&	0.883	&	0.878	&	0.847	&	0.800	&	0.745	&	0.657	\\
50	&	1.000	&	1.000	&	1.000	&	1.000	&	1.000	&	1.000	&	0.944	&	0.941	&	0.920	&	0.885	&	0.840	&	0.761	\\
\hline
	\end{tabular}
	%}
	\label{TAB:lrm_Eff}
\end{table}

\bigskip
Again, based on the general theory developed in Section \ref{sec4}, 
%and using the form of the influence function of the MDPDEs of $(\boldsymbol{\beta}, \phi)$ from \cite{Ghosh/Basu:2013},
one can readily check that the second order influence function of the proposed Wald-type tests 
at the true null distribution having parameters $\boldsymbol{\theta}_0 = (\boldsymbol{\beta}_0^T, \phi_0)^T$
under the present case simplifies to
\begin{eqnarray}
IF_{i_0}^{(2)}(t_{i_0}; W_{\tau}, \mathbf{\underline{F}}_{{\boldsymbol{\theta}}_0}) 
&=& \frac{2}{\phi_0} (1+2\tau)^{3/2}\left(t_{i_0}-\boldsymbol{x}_{i_0}^T\boldsymbol{\beta}_0\right)^2
e^{-\frac{\tau\left(t_{i_0}-\boldsymbol{x}_{i_0}^T\boldsymbol{\beta}_0\right)^2}{\phi_0}}
\left[\boldsymbol{x}_{i_0}^T(\boldsymbol{X}^T\boldsymbol{X})^{-1}\boldsymbol{C}_x(\boldsymbol{X}^T\boldsymbol{X})^{-1}\boldsymbol{x}_{i_0}
\right], \nonumber\\
%\label{EQ:IF2W0}\\
IF^{(2)}(\boldsymbol{t}; W_{\tau}, \mathbf{\underline{F}}_{{\boldsymbol{\theta}}_0}) 
&=& \frac{2}{\phi_0} (1+2\tau)^{3/2}\sum_{i=1}^n\left(t_{i}-\boldsymbol{x}_{i}^T\boldsymbol{\beta}_0\right)^2
e^{-\frac{\tau\left(t_{i}-\boldsymbol{x}_{i}^T\boldsymbol{\beta}_0\right)^2}{\phi_0}}
\left[\boldsymbol{x}_{i}^T(\boldsymbol{X}^T\boldsymbol{X})^{-1}\boldsymbol{C}_x(\boldsymbol{X}^T\boldsymbol{X})^{-1}\boldsymbol{x}_{i}
\right].\nonumber
%\label{EQ:IF2Wa}
\end{eqnarray}
Clearly, these influence functions depend on the values of the given fixed design-points in the direction 
of contamination. However, for any given finite design-points, they are bounded in contamination points $t_i$s 
for each $\tau>0$ and unbounded at $\tau=0$. We will explicitly examine their nature for some particular 
cases of design matrix; in particular, we consider the following four fixed designs:

\begin{figure}[!th]
	\centering
	%------------------------------------------------------
	\subfloat[Design 1, $i_0=10$]{
		\includegraphics[width=.3\textwidth] {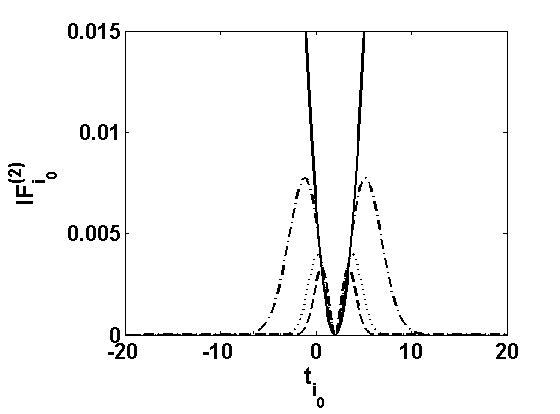}
		\label{fig:apower_30}}
	~ %--------------------------------------------------------------------
	\subfloat[Design 1, $i_0=40$]{
		\includegraphics[width=.3\textwidth] {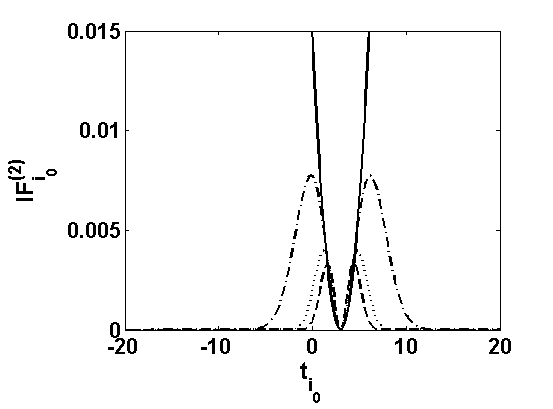}
		\label{fig:apower_302}}
	~ %--------------------------------------------------------------------
	\subfloat[Design 1, all directions]{
		\includegraphics[width=.3\textwidth] {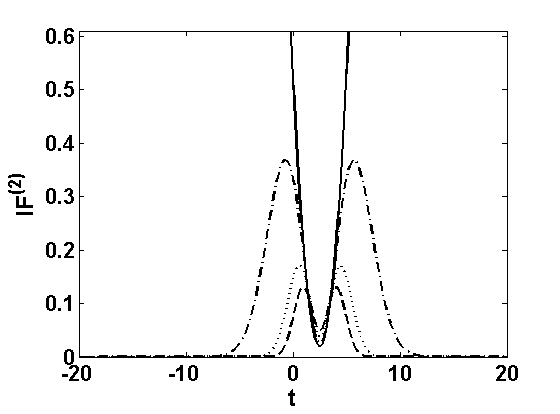}
		\label{fig:apower_303}}
	\\ %--------------------------------------------------------------------
	\subfloat[Design 2, $i_0=10$]{
		\includegraphics[width=.3\textwidth] {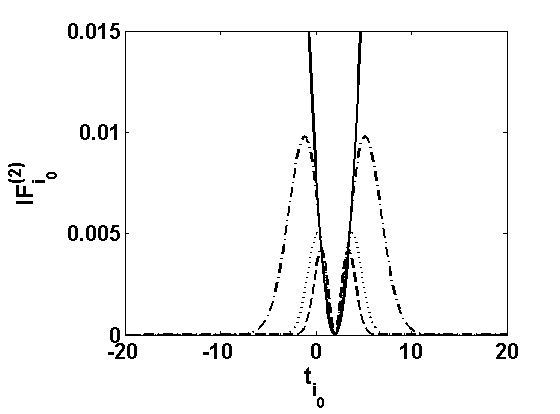}
		\label{fig:apower_30}}
	~ %--------------------------------------------------------------------
	\subfloat[Design 2, $i_0=40$]{
		\includegraphics[width=.3\textwidth] {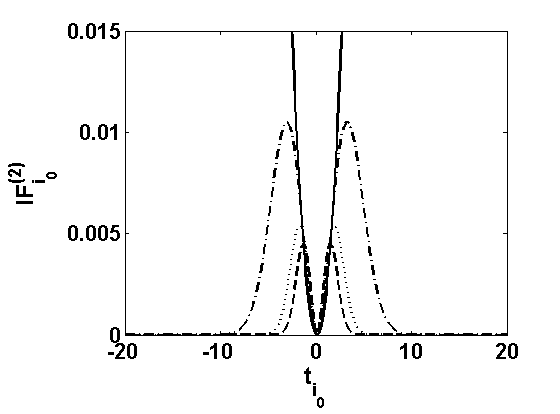}
		\label{fig:apower_302}}
	~ %--------------------------------------------------------------------
	\subfloat[Design 2, all directions]{
		\includegraphics[width=.3\textwidth] {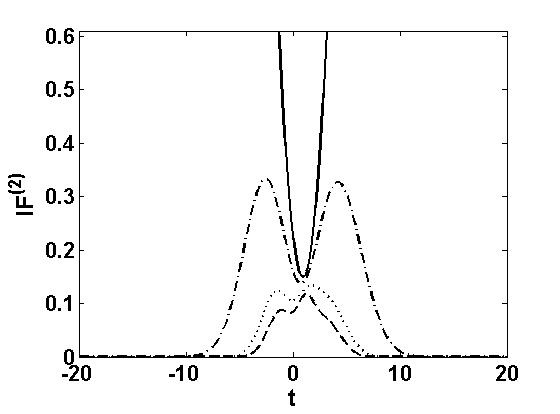}
		\label{fig:apower_303}}
	\\ %--------------------------------------------------------------------
	\subfloat[Design 3, $i_0=10$]{
		\includegraphics[width=.3\textwidth] {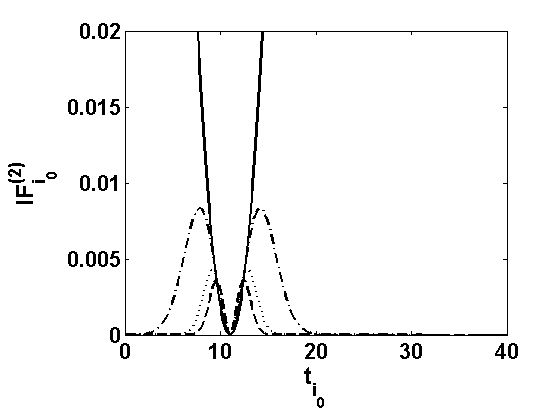}
		\label{fig:apower_30}}
	~ %--------------------------------------------------------------------
	\subfloat[Design 3, $i_0=40$]{
		\includegraphics[width=.3\textwidth] {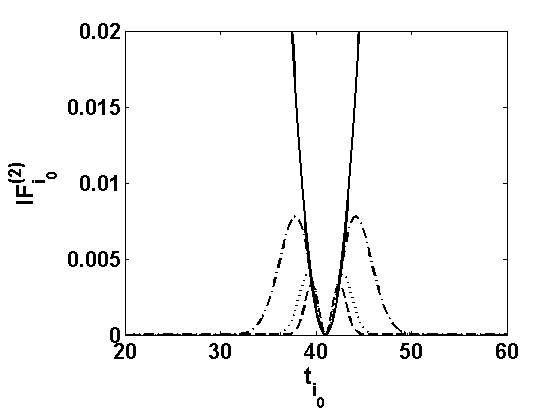}
		\label{fig:apower_302}}
	~ %--------------------------------------------------------------------
	\subfloat[$^*$ Design 3, all directions ]{
		\includegraphics[width=.3\textwidth] {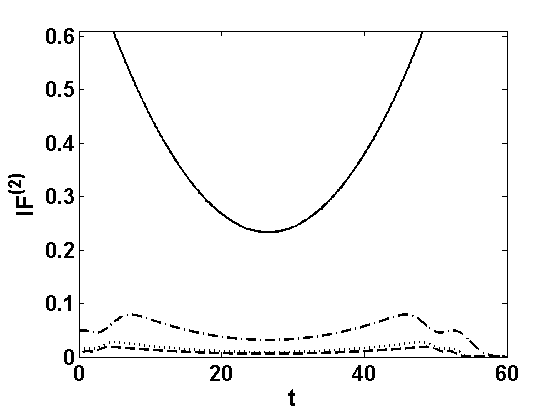}
		\label{FIG:1i}}
	\\ %--------------------------------------------------------------------
	\subfloat[Design 4, $i_0=10$]{
		\includegraphics[width=.3\textwidth] {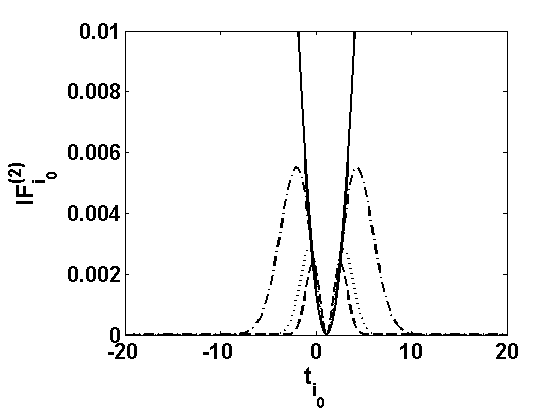}
		\label{fig:apower_30}}
	~ %--------------------------------------------------------------------
	\subfloat[Design 4, $i_0=40$]{
		\includegraphics[width=.3\textwidth] {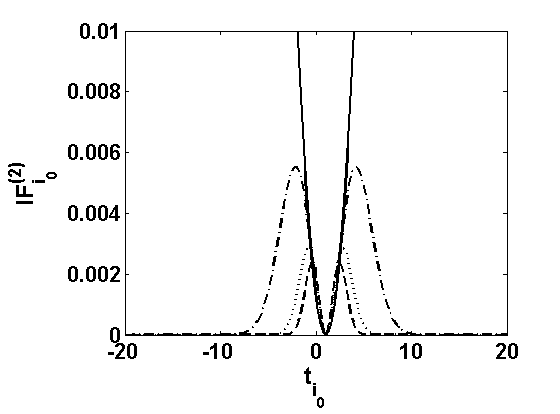}
		\label{fig:apower_302}}
	~ %--------------------------------------------------------------------
	\subfloat[Design 4, all directions]{
		\includegraphics[width=.3\textwidth] {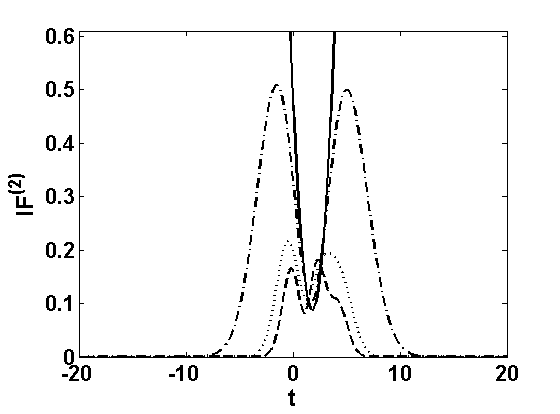}
		\label{fig:apower_303}}
	%----------------------------------------------------------------------------------------------
	\caption{Second order influence function of the proposed Wald-type test statistics for testing (\ref{EQ:lrm_hyp}) under the normal regression model with fixed designs 1 -- 4 and contamination in the direction $i_0=10, 40$ or in all directions 
		at $\boldsymbol{t}=t\boldsymbol{1}$  
		[solid line: $\tau=0$; dash-dotted line: $\tau=0.1$; dotted line: $\tau=0.3$; dashed line: $\tau=0.5$].
		~~~~~~~~~~~~~~~~~~~~~~~~~~~~~~~~~~~~~~~~~~~~~~~~~~~~~~~~~~~~~~~~~~~~
		$^*$ indicates that the values for $\tau=0$ (solid line) has been shown in multiple of $10^{-2}$ for this graph only.}
	\label{FIG:IF2_lrm}
\end{figure}

%\noindent
$$
\begin{array}{lll}
\mbox{Design 1:} & \boldsymbol{x}_i=(1, x_{i1})^T;~~x_{i1}=a, i=1,\ldots, n/2; x_{i1}=b, i=n/2+1,\ldots, n.
& \mbox{(two-point design)}\\
\mbox{Design 2:} & \boldsymbol{x}_i=(1, x_{i1})^T;~~{x}_{i1}, i=1,\ldots, n, 
\mbox{are pre-fixed iid observations from }N(\mu_x,\sigma_x^2).
& \mbox{(Fixed-Normal design)}\\
\mbox{Design 3:} & \boldsymbol{x}_i=(1, x_{i1})^T;~~x_{i1}=i~\mbox{for}~i=1,\ldots, n.
& \mbox{(Divergent design)} \\
\mbox{Design 4:} & \boldsymbol{x}_i=(1, x_{i1}, x_{i2})^T;~~x_{i1}=\frac{1}{i},~x_{i2}=\frac{1}{i^2}~\mbox{for}~i=1,\ldots, n.
& \mbox{(Convergent design)}
\end{array}
$$
Note that, the $\boldsymbol{C}_x$ matrix is finitely defined and is positive definite for the first two designs
with values 
$$\boldsymbol{C}_x=\begin{pmatrix}
\begin{array}{cc}
1 &\frac{1}{2}(a+b)\\ \frac{1}{2}(a+b) & \frac{1}{2}(a^2+b^2)
\end{array}
\end{pmatrix}
~~\mbox{and}~~\boldsymbol{C}_x=\begin{pmatrix}
\begin{array}{cc}
1 &\mu_x\\ \mu_x & \sigma_x^2+\mu_x^2
\end{array}
\end{pmatrix}
$$ respectively.
In our illustrations, we have taken $a=1$, $b=2$ in design 1 and $\mu_x=0$, $\sigma_x=1$ in Design 2
so that the first one is asymmetric and non-orthogonal but the second one is symmetric and orthogonal.
The design matrix for Designs 3 and 4 are positive definite for any finite sample sizes, 
but the corresponding $\boldsymbol{C}_x$ matrices have all elements except their $(1,1)$-th element as $\infty$ and $0$ respectively;
however,  for the computation of the above fixed sample IFs in these cases 
we can use the finite sample approximation of $\boldsymbol{C}_x$ by $\frac{1}{n}(\boldsymbol{X}^T\boldsymbol{X})$.
Figure \ref{FIG:IF2_lrm} presents the second order  IF of our test statistics for different contamination direction
under these four designs at the finite sample size $n=50$ with $\boldsymbol{\beta}_0=\boldsymbol{1}$, the vector of ones, $\phi_0=1$
and different values of $\tau$. The boundedness of these IFs at $\tau>0$ clearly indicates the robustness of our proposed 
Wald-type tests. Further, the (absolute) supremum of these IFs also decreases as $\tau$ increases 
which implies the increasing robustness of the proposal with increasing $\tau$. 
The extent of this increase in robustness for $\tau>0$ over $\tau=0$ becomes more prominent when contamination is in all directions
and/or the size of the fixed design matrix increases (an extreme increment is in the case of Figure \ref{FIG:1i}).

Noting that the LIF is always zero, we can next study the power influence function also.
In the present case, the PIF can be seen to have the form 
\begin{eqnarray}
PIF(\mathbf{t}; W_{n}, \mathbf{\underline{F}}_{{\boldsymbol{\theta}}_0} ) 
&=& K_{k}^{\ast}\left( \delta\right)\frac{2}{\phi_0} (1+2\tau)^{3/2}(1+\tau)^{-3/2}
\sum_{i=1}^n\left(t_{i}-\boldsymbol{x}_{i}^T\boldsymbol{\beta}_0\right)
e^{-\frac{\tau\left(t_{i}-\boldsymbol{x}_{i}^T\boldsymbol{\beta}_0\right)^2}{2\phi_0}}
\left[\boldsymbol{d}^{T}\boldsymbol{C}_x(\boldsymbol{X}^T\boldsymbol{X})^{-1}\boldsymbol{x}_{i}\right],\nonumber
\end{eqnarray}
where $\delta$ is as given by Equation (\ref{EQ:delta_lrm}).
Figure \ref{FIG:PIF_lrm} presents these PIFs for the above four designs with different $\tau$
at the finite sample size $n=50$ with $\boldsymbol{\beta}_0=\boldsymbol{1}$, $\phi_0=1$, 
$\boldsymbol{d}=10^{-2}\boldsymbol{\beta}$ and $5\%$ level of significance. 
Again, the power of the proposed Wald-type tests under the normal model seems to be robust for all $\tau>0$ 
and for all the fixed designs over the classical non-robust choice of $\tau=0$. 
Further, the extent of robustness increases as $\tau$ increases or the size of the design matrix decreases.

\begin{figure}[!th]
	\centering
	%------------------------------------------------------
	\subfloat[Design 1]{
		\includegraphics[width=.3\textwidth] {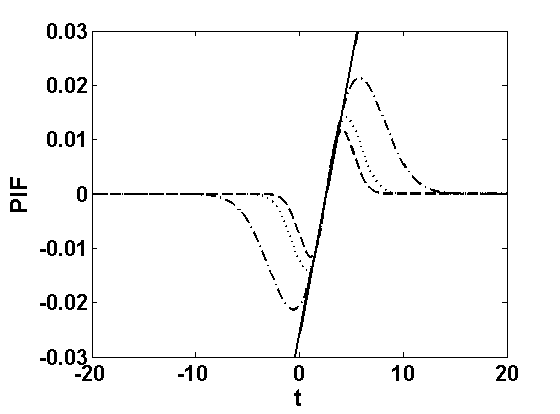}
		\label{fig:apower_30}}
	~ %--------------------------------------------------------------------
	\subfloat[Design 2]{
		\includegraphics[width=.3\textwidth] {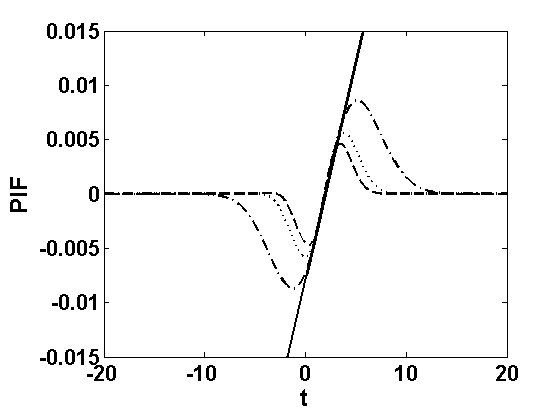}
		\label{fig:apower_302}}
	\\ %--------------------------------------------------------------------
	\subfloat[Design 3]{
		\includegraphics[width=.3\textwidth] {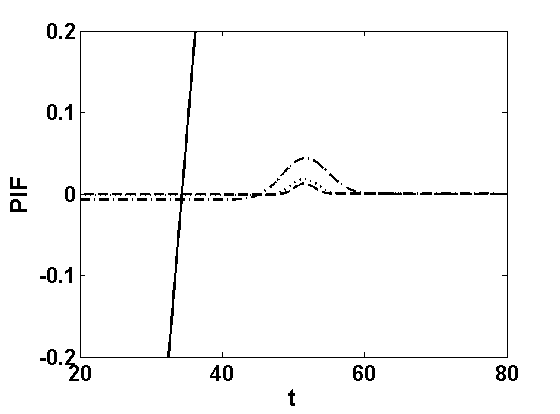}
		\label{fig:apower_302}}
	~ %--------------------------------------------------------------------
	\subfloat[Design 4]{
		\includegraphics[width=.3\textwidth] {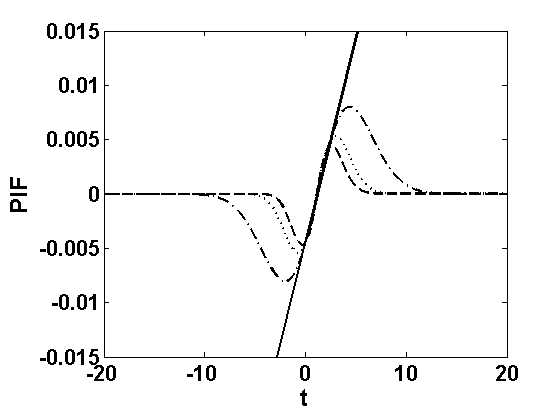}
		\label{fig:apower_303}}
	%----------------------------------------------------------------------------------------------
	\caption{Power influence function of the proposed Wald-type test statistics for testing (\ref{EQ:lrm_hyp}) under the normal regression model with fixed designs 1 -- 4 at contamination point $\boldsymbol{t}=t\boldsymbol{1}$
		[solid line: $\tau=0$; dash-dotted line: $\tau=0.1$; dotted line: $\tau=0.3$; dashed line: $\tau=0.5$].}
	\label{FIG:PIF_lrm}
\end{figure}

\subsection{Testing for individual regression coefficients in Poisson model for Count Data}

Let us consider another popular special class of GLM applicable to the analysis of count responses, 
namely the Poisson regression model. Here the response $y_i$ is assumed to have a Poisson distribution with 
mean $e^{\boldsymbol{x}_i^T\boldsymbol{\beta}}$ depending on the given predictor values $\boldsymbol{x}_i$.
In terms of the GLM notations of Section \ref{sec4}, the density in (\ref{A}) is then a Poisson density
with $\eta_i = \boldsymbol{x}_i^T\boldsymbol{\beta}$, known $\phi=1$ and the logarithmic link function. 
So, we can obtain the robust MDPDE $\widehat{\boldsymbol{\theta }}_{\tau}=\widehat{\boldsymbol{\beta }}_{\tau}$ 
of the regression parameter $\boldsymbol{\theta}=\boldsymbol{\beta}$ in this case following the general theory 
of Section \ref{sec4}. \cite{Ghosh/Basu:2016} studied these MDPDEs  $\widehat{\boldsymbol{\beta }}_{\tau}$
under the fixed-design Poisson regression model and their properties in detail with examples. 
In particular, in the notations of Section \ref{sec4}, we have estimating equations given only by (\ref{B}) with 
$K_{1}(y_{i},\boldsymbol{x}_{i}^{T}\boldsymbol{\beta },\phi ) = \left(y_{i} - \boldsymbol{x}_i^T\boldsymbol{\beta}\right)$,
and the required asymptotic variance matrix $\boldsymbol{\Sigma}_{\tau}$ can be obtained 
in terms of $\boldsymbol{\Omega}_{n,\tau }$ and $\boldsymbol{\Psi}_{n,\tau }$ as defined in (\ref{EQ:mat_betaO}).
%see \cite{Ghosh/Basu:2016} for more details.

Here, as our second illustration of the proposed Wald-type testing procedures, 
we will consider the problem of testing for the significance of any predictor (say, $h$-th predictor) in the model.
For a fixed integer $h$ between 1 to $k$, the corresponding hypothesis is given by 
\begin{equation}
H_{0}:\beta_h=0
\text{ versus }
H_{1}:{\beta_h } \neq 0,  
\label{EQ:Hyp_pr}
\end{equation}%
where $\beta_h$ denotes the $h$-th component of the regression vector $\beta$.
Clearly this important hypothesis in (\ref{EQ:Hyp_pr}) is a special case of the general linear hypotheses in (\ref{G})
with $r=1$, $\boldsymbol{L}^T$ being an $k$-vector with all entries zero except the $h$-th entry as 1 and ${l}_0=0$.
So, our proposed Wald-type test statistics for this problem, following the general theory of Section \ref{sec4},
can be simplified as 
\begin{equation}
W_{n}(\widehat{\boldsymbol{\beta }}_{\tau})
=\frac{n\widehat{{\beta }}_{h,\tau}^2}{\sigma_{hh,\tau}^2(\widehat{\boldsymbol{\beta }}_{\tau})},  
\label{EQ:TS_pr}
\end{equation}
where $\widehat{{\beta }}_{h,\tau}$ is the $h$-th entry of $\widehat{\boldsymbol{\beta }}_{\tau}$
denoting the MDPDE of $\beta_h$ and $\sigma^2_{hh,\tau}$ denote the $h$-th diagonal element 
of the asymptotic covariance matrix $\boldsymbol{\Sigma}_{\tau}$ at the null parameter values
denoting the null asymptotic variance of $\sqrt{n}\widehat{{\beta }}_{h,\tau}$.
Following Section \ref{sec4}, this test statistics in (\ref{EQ:TS_pr}) is asymptotically 
distributed as $\chi_1^2$ distribution and consistent at any fixed alternatives. 
Further, denoting the null parameter value as $\boldsymbol{\beta}_0$ having $h$-th entry $\beta_{h,0}=0$, 
the asymptotic distribution of the proposed Wald-type test statistics under the contiguous alternatives  
$$
H_{1,n} : \boldsymbol{\beta}_{n}~\mbox{with}~\beta_{h,n} =  n^{-1/2}d, \beta_{l,n} = \beta_{l,0}~\mbox{for}~l\neq h, 
d\in\mathbb{R} - \{{0}\}
$$  
is a non-central $\chi^2$ distribution with one degree of freedom and non-centrality parameter 
\begin{equation}
\delta=\frac{d^2}{\sigma_{hh,\tau}^2({\boldsymbol{\beta }}_{0})}. 
\label{EQ:delta_pr}
\end{equation}
Note that $\sigma_{hh,\tau}^2({\boldsymbol{\beta }}_{0})$ has no closed form expression in this case
but can be estimated numerically for any fixed sample size and any given design-matrix by 
$\widehat{\sigma}_{hh,\tau}^2({\boldsymbol{\beta }}_{0})$, the $h$-th diagonal entry of the matrix 
$$\boldsymbol{\Psi}_{n,\tau }^{-1}({\boldsymbol{\beta }}_{0}) \boldsymbol{\Omega}_{n,\tau }({\boldsymbol{\beta }}_{0})
\boldsymbol{\Psi}_{n,\tau }^{-1}({\boldsymbol{\beta }}_{0})$$ estimating $\boldsymbol{\Sigma}_{\tau}({\boldsymbol{\beta }}_{0})$.
Therefore, the effect of given design points can not be separated out explicitly from the form of 
asymptotic contiguous power based on this non-central distribution as was the case for previous normal model.
We again consider the four given designs 1--4 from Section \ref{sec5.1} and numerically compute the 
asymptotic contiguous power of the proposed Wald-type tests for testing (\ref{EQ:Hyp_pr}) 
%under the fixed design Poisson regression model 
with different values of $h$, $d$ and $\tau$ assuming $n=50$, 
$\beta_{l,0}=1$ for all $l\neq h$ and $5\%$ level of significance;
the results are shown in Table \ref{TAB:Eff_pr}. 
Once again, the power loss is not quite significant for any small positive values of $\tau$.
Also, we need larger values of $d$ to attain any fixed power by the proposed Wald-type test statistics 
with fixed tuning parameter whenever the values of the fixed design variables increases.

\begin{table}[h]
	\caption{Contiguous power of the proposed Wald-type test for testing (\ref{EQ:Hyp_pr}) under the Poisson regression model}
	\centering
	\resizebox{\textwidth}{!}{ %
	\begin{tabular}{l|l| cccccc| cccccc} \hline
Design	&	& \multicolumn{6}{|c}{$\tau$}& \multicolumn{6}{|c}{$\tau$}\\
	&	$d_x$	&	0	&	0.1	&	0.3	&	0.5	&	0.7	&	1 &	0		&	0.1	&	0.3	&	0.5	&	0.7	&	1	\\\hline
Design 1	&	& \multicolumn{6}{|c}{$h=1$}& \multicolumn{6}{|c}{$h=2$}\\\hline
	&	0	&	0.050	&	0.050	&	0.050	&	0.050	&	0.050	&	0.050	&	0.050	&	0.050	&	0.050	&	0.050	&	0.050	&	0.050	\\
&	2	&	0.200	&	0.199	&	0.189	&	0.178	&	0.167	&	0.152	&	0.378	&	0.375	&	0.355	&	0.331	&	0.308	&	0.276	\\
&	3	&	0.388	&	0.383	&	0.364	&	0.339	&	0.315	&	0.282	&	0.696	&	0.691	&	0.663	&	0.627	&	0.589	&	0.534	\\
&	5	&	0.796	&	0.792	&	0.766	&	0.730	&	0.692	&	0.634	&	0.985	&	0.984	&	0.978	&	0.968	&	0.954	&	0.926	\\
&	7	&	0.974	&	0.973	&	0.964	&	0.950	&	0.932	&	0.897	&	1.000	&	1.000	&	1.000	&	1.000	&	0.999	&	0.998	\\
&	10	&	1.000	&	1.000	&	1.000	&	0.999	&	0.999	&	0.996	&	1.000	&	1.000	&	1.000	&	1.000	&	1.000	&	1.000	\\
\hline
Design 2	&	& \multicolumn{6}{|c}{$h=1$}& \multicolumn{6}{|c}{$h=2$}\\\hline
	&	0	&	0.050	&	0.050	&	0.050	&	0.050	&	0.050	&	0.050	&	0.050	&	0.050	&	0.050	&	0.050	&	0.050	&	0.050	\\
&	1	&	0.140	&	0.139	&	0.132	&	0.124	&	0.117	&	0.110	&	0.320	&	0.316	&	0.300	&	0.280	&	0.261	&	0.234	\\
&	2	&	0.414	&	0.408	&	0.381	&	0.351	&	0.326	&	0.299	&	0.846	&	0.842	&	0.819	&	0.786	&	0.749	&	0.693	\\
&	3	&	0.743	&	0.735	&	0.700	&	0.658	&	0.619	&	0.574	&	0.994	&	0.994	&	0.991	&	0.985	&	0.977	&	0.959	\\
&	5	&	0.992	&	0.991	&	0.985	&	0.976	&	0.965	&	0.947	&	1.000	&	1.000	&	1.000	&	1.000	&	1.000	&	1.000	\\
&	7	&	1.000	&	1.000	&	1.000	&	1.000	&	1.000	&	0.999	&	1.000	&	1.000	&	1.000	&	1.000	&	1.000	&	1.000	\\
\hline
Design 3	&	& \multicolumn{6}{|c}{$h=1$}& \multicolumn{6}{|c}{$h=2$}\\\hline
	&	0	&	0.050	&	0.050	&	0.050	&	0.050	&	0.050	&	0.050	&	0.050	&	0.050	&	0.050	&	0.050	&	0.050	&	0.050	\\
&	0.01	&	1.000	&	1.000	&	1.000	&	1.000	&	1.000	&	1.000	&	0.057	&	0.056	&	0.056	&	0.056	&	0.055	&	0.055	\\
	&	0.05	&	1.000	&	1.000	&	1.000	&	1.000	&	1.000	&	1.000	&	0.221	&	0.219	&	0.209	&	0.196	&	0.183	&	0.166	\\
&	0.1	&	1.000	&	1.000	&	1.000	&	1.000	&	1.000	&	1.000	&	0.662	&	0.657	&	0.630	&	0.593	&	0.557	&	0.503	\\
&	0.2	&	1.000	&	1.000	&	1.000	&	1.000	&	1.000	&	1.000	&	0.997	&	0.997	&	0.996	&	0.993	&	0.988	&	0.976	\\
&	0.5	&	1.000	&	1.000	&	1.000	&	1.000	&	1.000	&	1.000	&	1.000	&	1.000	&	1.000	&	1.000	&	1.000	&	1.000	\\
\hline
Design 4	&	& \multicolumn{6}{|c}{$h=2$}& \multicolumn{6}{|c}{$h=3$}\\\hline
	&	0	&	0.050	&	0.050	&	0.050	&	0.050	&	0.050	&	0.050	&	0.050	&	0.050	&	0.050	&	0.050	&	0.050	&	0.050	\\
&	10	&	0.153	&	0.152	&	0.145	&	0.137	&	0.129	&	0.119	&	0.168	&	0.167	&	0.159	&	0.149	&	0.140	&	0.128	\\
&	20	&	0.459	&	0.455	&	0.431	&	0.402	&	0.373	&	0.333	&	0.510	&	0.506	&	0.479	&	0.445	&	0.412	&	0.366	\\
&	30	&	0.795	&	0.792	&	0.765	&	0.728	&	0.689	&	0.630	&	0.845	&	0.841	&	0.816	&	0.781	&	0.740	&	0.679	\\
&	50	&	0.996	&	0.996	&	0.994	&	0.990	&	0.983	&	0.968	&	0.999	&	0.999	&	0.998	&	0.995	&	0.992	&	0.981	\\
&	70	&	1.000	&	1.000	&	1.000	&	1.000	&	1.000	&	1.000	&	1.000	&	1.000	&	1.000	&	1.000	&	1.000	&	1.000	\\
\hline
	\end{tabular}}
	\label{TAB:Eff_pr}
\end{table}

\begin{figure}[!th]
	\centering
	%------------------------------------------------------
	\subfloat[Design 1, $i_0=10$]{
		\includegraphics[width=.3\textwidth] {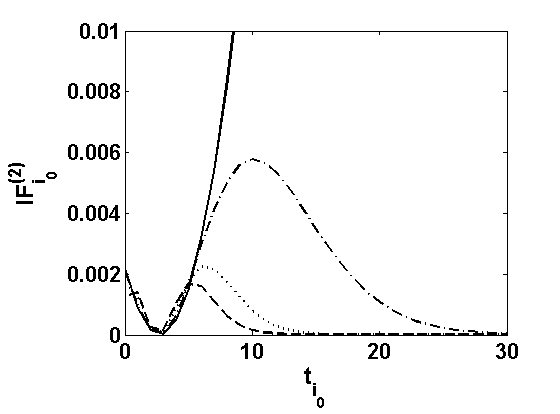}
		\label{fig:apower_30}}
	~ %--------------------------------------------------------------------
	\subfloat[Design 1, $i_0=40$]{
		\includegraphics[width=.3\textwidth] {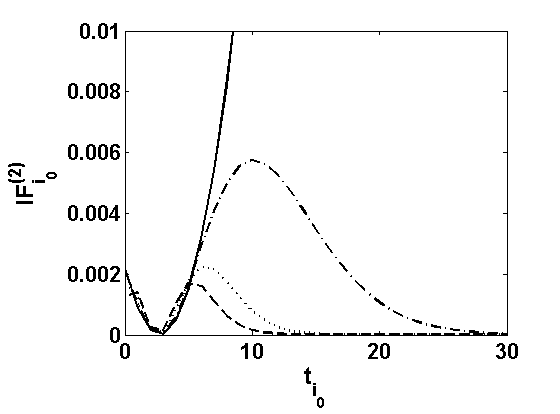}
		\label{fig:apower_302}}
	~ %--------------------------------------------------------------------
	\subfloat[Design 1, all directions]{
		\includegraphics[width=.3\textwidth] {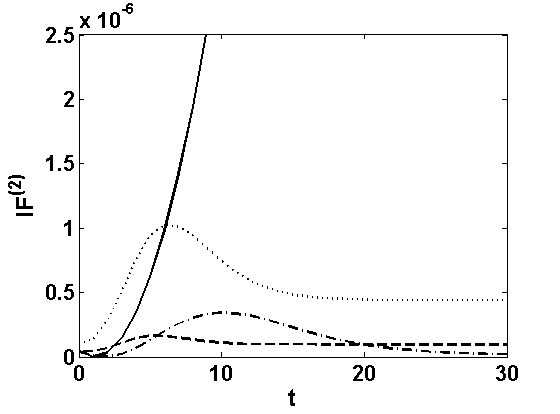}
		\label{fig:apower_303}}
	\\ %--------------------------------------------------------------------
	\subfloat[Design 2, $i_0=10$]{
		\includegraphics[width=.3\textwidth] {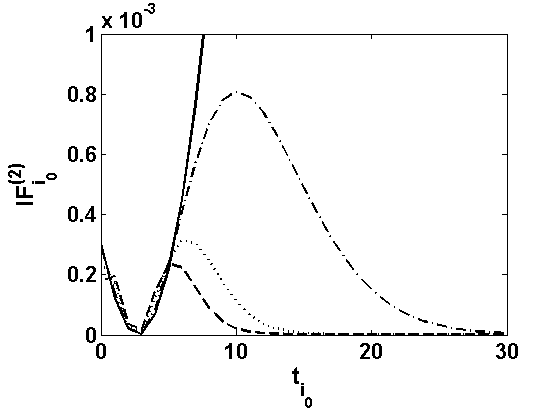}
		\label{fig:apower_30}}
	~ %--------------------------------------------------------------------
	\subfloat[Design 2, $i_0=40$]{
		\includegraphics[width=.3\textwidth] {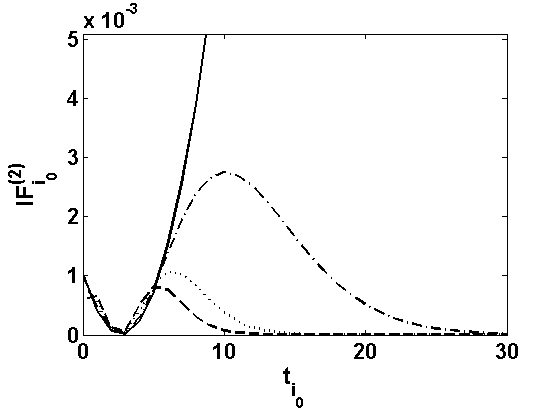}
		\label{fig:apower_302}}
	~ %--------------------------------------------------------------------
	\subfloat[Design 2, all directions]{
		\includegraphics[width=.3\textwidth] {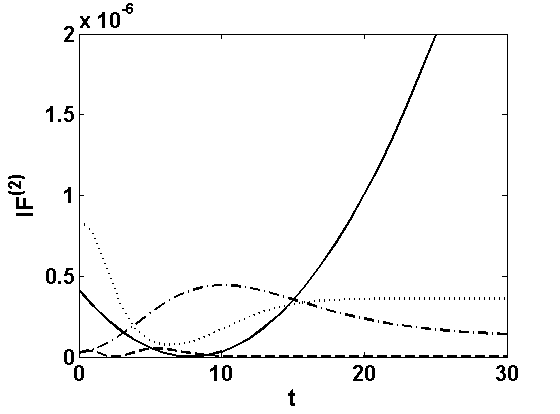}
		\label{fig:apower_303}}
	\\ %--------------------------------------------------------------------
	\subfloat[Design 3, $i_0=10$]{
		\includegraphics[width=.3\textwidth] {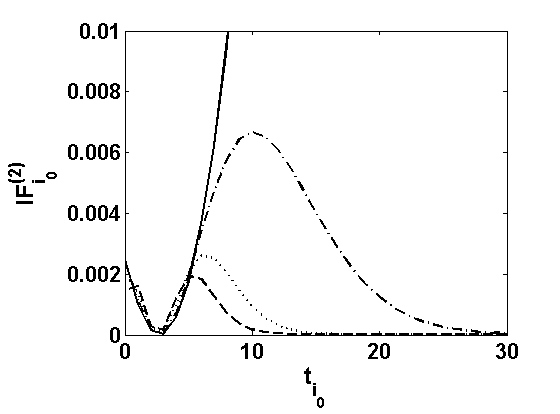}
		\label{fig:apower_30}}
	~ %--------------------------------------------------------------------
	\subfloat[Design 3, $i_0=40$]{
		\includegraphics[width=.3\textwidth] {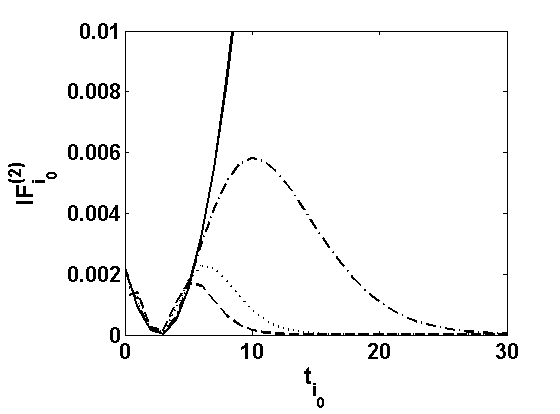}
		\label{fig:apower_302}}
	~ %--------------------------------------------------------------------
	\subfloat[Design 3, all directions ]{
		\includegraphics[width=.3\textwidth] {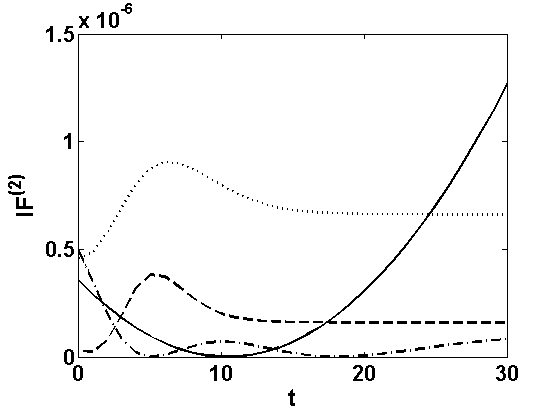}
		\label{FIG:1i}}
	\\ %--------------------------------------------------------------------
	\subfloat[Design 4, $i_0=10$]{
		\includegraphics[width=.3\textwidth] {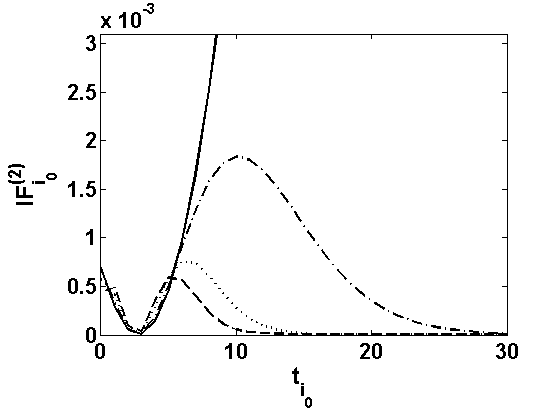}
		\label{fig:apower_30}}
	~ %--------------------------------------------------------------------
	\subfloat[Design 4, $i_0=40$]{
		\includegraphics[width=.3\textwidth] {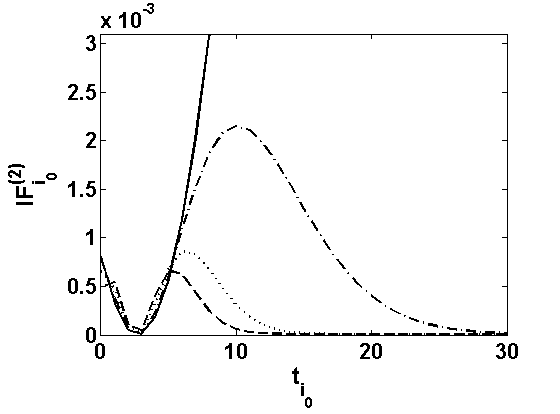}
		\label{fig:apower_302}}
	~ %--------------------------------------------------------------------
	\subfloat[Design 4, all directions]{
		\includegraphics[width=.3\textwidth] {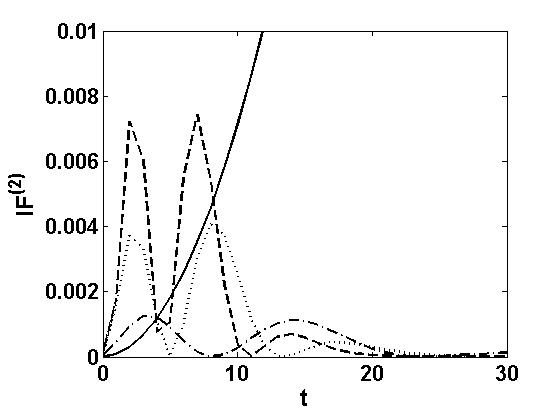}
		\label{fig:apower_303}}
	%----------------------------------------------------------------------------------------------
	\caption{Second order influence function of the proposed Wald-type test statistics for testing (\ref{EQ:Hyp_pr}) 
		under Poisson regression model with fixed designs 1 -- 4 and contamination in the direction $i_0=10, 40$ or in all directions 
		at $\boldsymbol{t}=t\boldsymbol{1}$  
		[solid line: $\tau=0$; dash-dotted line: $\tau=0.1$; dotted line: $\tau=0.3$; dashed line: $\tau=0.5$].}
	\label{FIG:IF2_pr}
\end{figure}

Further, the robustness of our proposed Wald-type tests can also be verified by examining the 
second order influence function of the test statistics and the power influence functions
from their general expressions derived in Section \ref{sec4}.
However, in the present case of Poisson regression model, we cannot have more simplified  explicit expressions for them
except for the particular form of $K_{1}(y_{i},\boldsymbol{x}_{i}^{T}\boldsymbol{\beta },\phi )$ as defined above and no 
$K_{2}(y_{i},\boldsymbol{x}_{i}^{T}\boldsymbol{\beta },\phi )$;
but they can be easily computed numerically for any given fixed design.
We again present the numerical values of the second order influence functions of the proposed Wald-type test statistics 
for testing significance of the first slope parameter $\beta_2$ ($h=2$ in (\ref{EQ:Hyp_pr})) at different values of $\tau$
under the four designs considered in Section \ref{sec5.1} with $n=50$ and $\beta_{l,0}=1$ for all $l\neq 2$;
these are presented in Figures \ref{FIG:IF2_pr}.
The redescending nature of all the influence functions 
with increasing $\tau$ is again quite clear from the figures, 
which indicates the increasing robustness of our proposed Wald-types tests as $\tau>0$ increases 
over the non-robust choice $\tau=0$ (having unbounded IFs).
The nature of the power influence functions in this case are also seen to be very similar 
implying the robustness of our proposal at $\tau>0$; so we skip them for brevity.

%\subsection{Testing under Logistic Regression model for Binary Data}

%\section{Simulation Study\label{sec6}}
%\section{Discussions\label{sec7}}
\newpage
\section{Concluding remarks and the Choice of $\tau$\label{sec7}}

We have proposed a robust parametric hypothesis testing approach for general non-homogeneous observations
involving a common model parameter. The test statistics have been constructed by generalizing the 
Wald test statistics using the robust minimum density power divergence estimator (with parameter $\tau\geq 0$)
of the underlying common parameter in place of its non-robust maximum likelihood estimator. 
The properties of the proposed test have been studied for both simple and composite hypotheses under general non-homogeneous set-up
and applied to the cases of several fixed design GLMs. 
In particular, it has been observed that the proposed tests have simple chi-square asymptotic limit under null hypothesis
in contrast to the linear combination of chi-square limit for the robust tests of \cite{Ghosh/Basu:2015}
making its application much easier under complex models. Also, the tests are always consistent at any fixed alternatives
and have bounded second order influence functions of the test statistics and bounded power influence functions for $\tau>0$
implying robustness of our proposal.

Further, in each of the examples considered, we have seen that the asymptotic power of the proposed Wald-type tests
under any contiguous alternatives as well as the extent of robustness depends on the tuning parameter $\tau$.
In particular, as $\tau$ increases, the contiguous power decreases slightly from its highest value at $\tau=0$
corresponding to the non-robust classical MLE based Wald test but the robustness increases significantly.
Thus, the tuning parameter $\tau$ yields a trade-off between asymptotic contiguous power and robustness of these Wald-type tests;
note the similarity with the trade-off between asymptotic efficiency and robustness of the underlying MDPDE 
as studied in \cite{Ghosh/Basu:2013}. In fact one can explicitly examine, from the theoretical results derived here,
that the dependence of the power and robustness of the proposed Wald-type tests 
comes directly through the similar dependence of the efficiency and robustness of the MDPDE used 
in constructing the test statistics. Hence a proper choice of the tuning parameter $\tau$ balancing  
the asymptotic power and robustness can be equivalently obtained by balancing the corresponding trade-off
for the underlying MDPDE. This latter problem under the non-homogeneous set-up has been proposed and studied by 
\cite{Ghosh/Basu:2013,Ghosh/Basu:2015a,Ghosh/Basu:2016}, where it is proposed that 
a data-driven estimate of the mean square error of the MDPDE be minimized 
to obtain the optimum tuning parameter for any given practical dataset.
The same optimum $\tau$ can also be used as well for applying our proposed Wald-type tests for any practical hypothesis testing problems.
However, more detailed investigation on this issue could ba an interesting future research work.

Another possible future extension of the present paper will be to construct similar robust testing procedures 
for two independent samples of non-homogeneous data. This problem is of high practical relevance as
one can then use the construction to test between two regression lines from fixed design clinical trials
occurring frequently in medical sciences and epidemiology. 
The corresponding problem with homogeneous sample has been recently tackled by \cite{Ghosh/etc:2016b}
which should be extended to the cases with non-homogeneous data and fixed-design regressions as in the present paper.
We hope to pursue some such extensions in future.

\appendix
\section{Assumptions} 
\label{APP:cond}
\renewcommand{\theenumi}{(A\arabic{enumi})}
\renewcommand\labelenumi{(A\arabic{enumi})}
\noindent\textbf{
Assumptions required for Asymptotic distributions of the MDPDE under non-homogeneous data \citep{Ghosh/Basu:2013}:}
\begin{enumerate}
	\item\label{as1} For all $i=1, \ldots, n$, the support of model distribution given by
	$\chi= \{y|f_{i,\boldsymbol{\theta}}(y) > 0\}$ is independent of $i$ and $\boldsymbol{\theta}$
	and is the same as the support of true distribution $G_i$ for each $i$.
	
	\item\label{as2} There exists an open subset $\omega\subseteq\Theta$
	that contains the best fitting parameter $\boldsymbol{\theta}^g=\boldsymbol{T}_{\tau}(\underline{\boldsymbol{G}})$
	and, at each $\boldsymbol{\theta}\in \omega$, the model density $f_{i,\boldsymbol{\theta}}(y)$ 
	is thrice continuously differentiable with respect to $\theta$ for almost all $y \in\chi$ and all $i =1,\ldots, n$.

	\item\label{as3} $\int f_{i,\boldsymbol{\theta}}(y)^{1+\alpha}dy$ and 
	$\int f_{i,\boldsymbol{\theta}}(y)^{\alpha}g_i(y)dy$ can be differentiated three times with respect to $\theta$, 
	and the derivatives can be taken under the integral sign for any $i =1,\ldots, n$.
	
	\item\label{as4} The matrix $\boldsymbol{J}_{i,\tau}$ is positive definite for any $i=1,\ldots, n$, and, 
	%\lambda_0 = 
	$\displaystyle\inf_n[ \min~\mbox{eigenvalue of } \boldsymbol{\Psi}_{n,\tau} ] > 0.$
	
	\item\label{as5} Define $V_{i,\boldsymbol{\theta}}(y)=\left[	\int f_{i,\boldsymbol{\theta}}(y)^{1+\alpha} dy 
	- \left(1+\frac{1}{\alpha}\right)f_{i,\boldsymbol{\theta}}(y)^\alpha\right]$.
	For all $\boldsymbol{\theta}\in \Theta$, $i=1, \ldots, n$, and all $j, h, l=1, \ldots, k$,
	the $(j,h,l)$-th (third order) partial derivatives of $V_{i,\boldsymbol{\theta}}(y)$ is bounded in absolute value 
	by some functions $M_{jhl}^{(i)}(y)$ satisfying
	$\displaystyle\frac{1}{n} \sum_{i=1}^{n} E_{g_i} \left[ M_{jhl}^{(i)}(Y)\right] =O(1)$.
		
	\item\label{as6} For all $j, h =1, \ldots, k$, define 
	$\boldsymbol{N}_{ijh}^{(1)}(y) = \nabla_{j}V_{i,\boldsymbol{\theta}}(y)$ and
	$\boldsymbol{N}_{ijh}^{(2)}(y) = \nabla_{jh}V_{i,\boldsymbol{\theta}}(y) - E_{g_i}(\nabla_{jh}V_{i,\boldsymbol{\theta}}(y))$.
	Then, we have
	%e10 ###
	%
	\begin{align}
	&\lim_{N\rightarrow\infty} \sup_{n>1} \Biggl\{ \frac{1}{n} \sum_{i=1}^n E_{g_i} |\boldsymbol{N}_{ijh}^{(l)}(y)| 
	I(|\boldsymbol{N}_{ijh}^{(l)}(y))| > N)\Biggr\} = 0, ~~l=1, 2.\nonumber
	\end{align}
	
	%
%	where $I(B)$ denotes the indicator variable of the event $B$.
	
	\item\label{as7} For any $\epsilon> 0$, 
	%e12 ###
	%
	\begin{equation}
	\lim_{n\rightarrow\infty} \left\{ \frac{1}{n} \sum_{i=1}^n E_{g_i}
	\left[||\boldsymbol\Omega_n^{-1/2} \nabla V_{i,\boldsymbol{\theta}}(y)||^2 
	I(||\boldsymbol{\Omega}_n^{-1/2} \nabla V_{i,\boldsymbol{\theta}}(y)|| > \epsilon\sqrt{n})\right]\right\} = 0.
	\label{EQ:A7}
	\end{equation}
\end{enumerate}

\bigskip
\renewcommand{\theenumi}{(R\arabic{enumi})}
\renewcommand\labelenumi{(R\arabic{enumi})}
\noindent\textbf{
	Assumptions required for Asymptotic distributions of the MDPDE under Normal Fixed-Design Linear Model \citep{Ghosh/Basu:2013}:}
\\
\noindent
The values of given design point $\boldsymbol{x}_i=(x_{1i},\ldots,x_{ki})^T$ are such that
\begin{enumerate}
\item\label{asR1}  $\displaystyle\sup_{n>1}\max_{1\le i\le n} ~ |x_{ji}| = O(1)$, 
$\displaystyle\sup_{n>1}\max_{1\le i\le n} ~ |x_{ji}x_{li}|= O(1)$,
and $\displaystyle\frac{1}{n} \sum_{i=1}^n ~ |x_{ji}x_{li}x_{hi}| = O(1)$,
 for all $j, l, h = 1, \ldots,k.$
%e34 ###
%
%\begin{equation}
%\sup_{n>1}\max_{1\le i\le n} ~ |x_{ji}| = O(1), ~~ \sup
%_{n>1}\max_{1\le i\le n} ~ |x_{ji}x_{li}|= O(1),~~\frac{1}{n} \sum_{i=1}^n ~ |x_{ji}x_{li}x_{hi}| = O(1),
%~\mbox{ for all }j, l, h = 1, \ldots,k.\nonumber
%%\label{EQ:R1}
%\end{equation}

\item\label{asR2} 
%The given design matrix $\boldsymbol{X}$ satisfies
$\displaystyle\inf_n \left[ \min~\mbox{eigenvalue of } \frac{1}{n}(\boldsymbol{X}^T \boldsymbol{X})\right] > 0.$
%e36 ###
%
%\begin{equation}
%\inf_n~~ [ \min~\mbox{eigenvalue of }~ \frac{1}{n}(\boldsymbol{X}^T \boldsymbol{X})] > 0. \nonumber
%%\label{EQ:R3}
%\end{equation}
%
%which also implies that the matrix $X$ has full column rank, and
%%e37 ###
%%
%\begin{equation}
%n \max_{1\le i\le n} ~ [x_i^T(X^T X)^{-1}x_i] = O(1). \label{EQ:R4}
%\end{equation}
%
\end{enumerate}

\end{document}